\documentclass[a4paper,11pt]{article}
\pdfoutput=1

\usepackage[margin=1in]{geometry}
\usepackage[utf8]{inputenc}
\usepackage[T1]{fontenc}
\usepackage[backend=biber,maxbibnames=10,maxcitenames=4,maxalphanames=4,style=alphabetic,isbn=false,giveninits=true]{biblatex}
\usepackage{hyperref}
\usepackage{color}
\usepackage{amsmath,amssymb,amsthm}
\usepackage{epsfig}
\usepackage{graphicx}

\hypersetup{
	hidelinks=false,
	draft=false,
	colorlinks=false, 
	linktoc=all 
}

\definecolor{cherry}{rgb}{0.9,.1,.2}
\definecolor{orange}{rgb}{1,0.5,0}

\usepackage{tikz}
\usetikzlibrary{backgrounds,calc,positioning,decorations.pathreplacing,decorations.markings,through,shapes}

\numberwithin{equation}{section}

\usepackage{comment} 
\includecomment{note}
\specialcomment{note}
    {\begingroup\color{blue}}{\endgroup} 
\includecomment{anne}
\specialcomment{anne}
    {\begingroup\color{orange}}{\endgroup} 
\excludecomment{note}

\usepackage{dsfont} 
\usepackage{cancel} 
\usepackage{soul} 
\usepackage{calc}
\usepackage{xifthen} 

\usepackage[nameinlink]{cleveref} 

\usepackage{thmtools}

\usepackage{braket}

\usepackage{comment} 
\includecomment{note}
\specialcomment{note}
    {\begingroup\color{blue}}{\endgroup} 
\includecomment{anne}
\specialcomment{anne}
    {\begingroup\color{orange}}{\endgroup} 
\excludecomment{note}

\usepackage{subcaption} 

\usepackage{authblk}

\declaretheorem[
style=plain,
name=Theorem,
numberwithin=section
]{theorem}
\declaretheorem[
style=plain,
name=Proposition,
numberlike=theorem
]{proposition}
\declaretheorem[
style=plain,
name=Lemma,
numberlike=theorem
]{lemma}
\declaretheorem[
style=plain,
name=Corollary,
numberlike=theorem
]{corollary}
\declaretheorem[
style=definition,
name=Definition,
numberlike=theorem
]{definition}
\declaretheorem[
style=definition,
name=Example,
numberlike=theorem,
qed=$\triangle$
]{example}

\let\deg\relax
\let\u\relax

\newcommand{\Om}[1]{\ensuremath{\Omega_{#1}}}
\newcommand{\N}[1]{\ensuremath{\mathcal{N}=#1}}
\newcommand{\CBI}{\ensuremath{\mathbb{C}B_I}}

\newcommand{\CB}[1]{\ensuremath{\mathbb{C}B_{I,#1}}}
\newcommand{\RB}[1]{\ensuremath{\mathbb{R}B_{I,#1}}}
\newcommand{\su}[1]{\ensuremath{\mathfrak{su}(#1)}}
\newcommand{\sut}[2][]{\ensuremath{\widehat{\su{2}^{#2}_{#1}}}}
\newcommand{\u}[1]{\ensuremath{\mathfrak{u}(#1)}}
\newcommand{\Ag}{\ensuremath{{A_\gamma}}}
\newcommand{\Agt}{\ensuremath{{\tilde{A}_\gamma}}}
\newcommand{\Aqu}{\ensuremath{{A_{QU}}}}
\renewcommand{\th}{\textsuperscript{th}}

\DeclareMathOperator{\Tr}{Tr}
\DeclareMathOperator{\STr}{STr}

\DeclareMathOperator{\Ch}{Ch}
\DeclareMathOperator{\SCh}{SCh}

\DeclareMathOperator{\SDet}{SDet}
\DeclareMathOperator{\deg}{deg}

\makeatletter
\newcommand{\syt}[2][0.5]{
		\raisebox{0.5ex}{\begin{tikzpicture}[scale=#1,baseline={(current bounding box.center)}]
			\@ifundefined{c@rownumber}
			{
				\newcounter{rownumber}
			}
			{
				\setcounter{rownumber}{0}
			}
			\foreach \row in #2
			{
				\ifnum\row=1
				{
					\draw (1,-\value{rownumber}) rectangle (2,-\value{rownumber}+1);
					\draw (1,-\value{rownumber}) -- (2,-\value{rownumber}+1);
				}
				\else
				{
					\foreach \x in {1,2,...,\row}
					{
						\draw (\x,-\value{rownumber}) rectangle (\x + 1,-\value{rownumber}+1);
						\draw (\x,-\value{rownumber}) -- (\x + 1,-\value{rownumber}+1);
					}
				}
				\fi
				\stepcounter{rownumber}
			}
			\useasboundingbox ([shift={(2mm,2mm)}]current bounding box.north east) rectangle ([shift={(-2mm,-2mm)}]current bounding box.south west);
		\end{tikzpicture}}
}

\newcommand{\csyt}[2][0.5]{
		\raisebox{0.5ex}{\begin{tikzpicture}[scale=#1,baseline={(current bounding box.center)}]
			\@ifundefined{c@rownumber}
			{
				\newcounter{rownumber}
			}
			{
				\setcounter{rownumber}{0}
			}
			\foreach \row in #2
			{
				\ifnum\row=1
				{
					\draw (1,-\value{rownumber}) rectangle (2,-\value{rownumber}+1);
					\draw (1,-\value{rownumber}) -- (2,-\value{rownumber}+1);
					\filldraw[fill=black] (1.5,-\value{rownumber}+0.5) circle (0.075cm);
				}
				\else
				{
					\foreach \x in {1,2,...,\row}
					{
						\draw (\x,-\value{rownumber}) rectangle (\x + 1,-\value{rownumber}+1);
						\draw (\x,-\value{rownumber}) -- (\x + 1,-\value{rownumber}+1);
						\filldraw[fill=black] (\x + 0.5,-\value{rownumber}+0.5) circle (0.075cm);
					}
				}
				\fi
				\stepcounter{rownumber}
			}
			\useasboundingbox ([shift={(2mm,2mm)}]current bounding box.north east) rectangle ([shift={(-2mm,-2mm)}]current bounding box.south west);
		\end{tikzpicture}}
}

\newcommand{\yt}[2][0.5]{
		\raisebox{0.5ex}{\begin{tikzpicture}[scale=#1,baseline={(current bounding box.center)}]
			\@ifundefined{c@rownumber}
			{
				\newcounter{rownumber}
			}
			{
				\setcounter{rownumber}{0}
			}
			\foreach \row in #2
			{
				\ifnum\row=1
				{
					\draw (1,-\value{rownumber}) rectangle (2,-\value{rownumber}+1);
				}
				\else
				{
					\foreach \x in {1,2,...,\row}
					{
						\draw (\x,-\value{rownumber}) rectangle (\x + 1,-\value{rownumber}+1);
					}
				}
				\fi
				\stepcounter{rownumber}
			}
			\useasboundingbox ([shift={(2mm,2mm)}]current bounding box.north east) rectangle ([shift={(-2mm,-2mm)}]current bounding box.south west);
		\end{tikzpicture}}
}

\newcommand{\nsymsyt}[2][0.5]{
	\ensuremath{
	\raisebox{0.5ex}{\begin{tikzpicture}[scale=#1,baseline={(0,0.25)}]
		\draw (0,0) rectangle (1,1);
		\draw (1,0) rectangle (3,1);
		\draw (3,0) rectangle (4,1);
		\draw (0,0) -- (1,1);
		\draw (3,0) -- (4,1);
		\draw [thick,decorate,decoration={brace,amplitude=5pt},yshift=4pt] (0,1)  -- (4,1) 
		   node [black,midway,above,yshift=3pt] {\scriptsize $#2$};
	\end{tikzpicture}}
	}
}

\newcommand{\nsymyt}[2][0.5]{
	\ensuremath{
	\raisebox{0.5ex}{\begin{tikzpicture}[scale=#1,baseline={(0,0.25)}]
		\draw (0,0) rectangle (1,1);
		\draw (1,0) rectangle (3,1);
		\draw (3,0) rectangle (4,1);
		\draw [thick,decorate,decoration={brace,amplitude=5pt},yshift=4pt] (0,1)  -- (4,1) 
		   node [black,midway,above,yshift=3pt] {\scriptsize $#2$};
	\end{tikzpicture}}
	}
}

\newcommand{\msymnantisyt}[3][0.5]{
	\ensuremath{
	\raisebox{0.5ex}{\begin{tikzpicture}[scale=#1,baseline={(current bounding box.center)}]
		\draw (0,0) rectangle (1,1);
		\draw (1,0) rectangle (3,1);
		\draw (3,0) rectangle (4,1);
		\draw (0,-2) rectangle (1,0);
		\draw (0,-3) rectangle (1,-2);
		\draw (0,0) -- (1,1);
		\draw (3,0) -- (4,1);
		\draw (0,-3) -- (1,-2);
		\draw [thick,decorate,decoration={brace,amplitude=5pt},yshift=4pt] (0,1)  -- (4,1) 
		   node [black,midway,above,yshift=3pt] {\scriptsize $#2$};
   		\draw [thick,decorate,decoration={brace,amplitude=5pt},xshift=-4pt] (0,-3) -- (0,1)
   		   node [black,midway,left,xshift=-3pt] {\scriptsize $#3$};
	\end{tikzpicture}}
	}
}

\newcommand{\nonmaxeccsyt}[3][0.5]{
	\ensuremath{
	\raisebox{0.5ex}{\begin{tikzpicture}[scale=#1,baseline={(current bounding box.center)}]
		\draw (0,0) rectangle (1,1);
		\draw (1,0) rectangle (2,1);
		\draw (2,0) rectangle (4,1);
		\draw (4,0) rectangle (5,1);
		\draw (0,-1) rectangle (1,0);
		\draw (0,-3) rectangle (1,-1);
		\draw (0,-4) rectangle (1,-3);
		\draw (1,-1) rectangle (2,0);
		\draw (0,0) -- (1,1);
		\draw (1,0) -- (2,1);
		\draw (4,0) -- (5,1);
		\draw (0,-1) -- (1,0);
		\draw (0,-4) -- (1,-3);
		\draw (1,-1) -- (2,0);
		\draw [thick,decorate,decoration={brace,amplitude=5pt},yshift=4pt] (2,1)  -- (5,1) 
		   node [black,midway,above,yshift=3pt] {\scriptsize $#2$};
   		\draw [thick,decorate,decoration={brace,amplitude=5pt},xshift=-4pt] (0,-4) -- (0,-1)
   		   node [black,midway,left,xshift=-3pt] {\scriptsize $#3$};
	\end{tikzpicture}}
	}
}

\newcommand{\zerosupertableau}[5][0.5]{
	\ensuremath{
	\raisebox{0.5ex}{\begin{tikzpicture}[scale=#1,baseline={(current bounding box.center)}]
		\draw (0,0) rectangle (1,1);
		\draw (1,0) rectangle (2,1);
		\draw (2,0) rectangle (3,1);
		\draw (0,-1) rectangle (1,0);
		\draw (1,-1) rectangle (2,0);
		\draw (2,-1) rectangle (3,0);
		\draw (0,-2) rectangle (1,-1);
		\draw (1,-2) rectangle (2,-1);
		\draw (2,-2) rectangle (3,-1);
		\draw (3,-1) rectangle (6,1);
		\draw (0,-5) rectangle (2,-2);
		\draw (6,0) rectangle (8,1);
		\draw (8,0) rectangle (9,1);
		\draw (0,-7) rectangle (1,-5);
		\draw (0,-8) rectangle (1,-7);
		\foreach \x in {0,1,2}
			\foreach \y in {0,-1,-2}
				\draw (\x,\y) -- (\x +1,\y +1);
		\draw (8,0) -- (9,1);
		\draw (0,-8) -- (1,-7);
		\draw [thick,decorate,decoration={brace,amplitude=5pt},yshift=4pt] (3,1)  -- (6,1) 
		   node [black,midway,above,yshift=3pt] {\scriptsize $#2$};
   		\draw [thick,decorate,decoration={brace,amplitude=5pt},yshift=4pt] (6,1)  -- (9,1) 
   		   node [black,midway,above,yshift=3pt] {\scriptsize $#3$};
   		\draw [thick,decorate,decoration={brace,amplitude=5pt},xshift=-4pt] (0,-5) -- (0,-2)
   		   node [black,midway,left,xshift=-3pt] {\scriptsize $#4$};
   		\draw [thick,decorate,decoration={brace,amplitude=5pt},xshift=-4pt] (0,-8) -- (0,-5)
   		   node [black,midway,left,xshift=-3pt] {\scriptsize $#5$};
	\end{tikzpicture}}
	}
}

\newcommand{\genericsupertableau}[5][0.5]{
	\ensuremath{
	\raisebox{0.5ex}{\begin{tikzpicture}[scale=#1,baseline={(current bounding box.center)}]
	\ifthenelse{\equal{#4}{0}}{\def\singlecolstarty{-1}}{\def\singlecolstarty{-4}}
	\ifthenelse{\equal{#2}{0}}{\def\singlerowstartx{2}}{\def\singlerowstartx{5}}
		\ifthenelse{\(\equal{#4}{0} \AND \equal{#5}{0}\) \OR \(\equal{#2}{0} \AND \equal{#3}{0}\)}{\ifthenelse{\equal{#4}{0}}{
		\draw (\singlerowstartx-3,0) -- (\singlerowstartx-2,0);
		\draw (\singlerowstartx-2,-1) -- (\singlerowstartx-2,1);}{
		\draw (0,\singlecolstarty+2) -- (2,\singlecolstarty+2);
		\draw (1,\singlecolstarty+2) -- (1,\singlecolstarty+3);}}{
		\draw (0,0) rectangle (1,1);
		\draw (1,0) rectangle (2,1);
		\draw (0,-1) rectangle (1,0);
		\draw (1,-1) rectangle (2,0);}
		\ifthenelse{\equal{#2}{0}}{}{
		\draw (2,-1) rectangle (5,1);}
		\ifthenelse{\equal{#4}{0}}{}{
		\draw (0,-4) rectangle (2,-1);}
		\ifthenelse{\equal{#3}{0}}{}{
		\draw (\singlerowstartx,0) rectangle (\singlerowstartx+2,1);
		\draw (\singlerowstartx+2,0) rectangle (\singlerowstartx+3,1);}
		\ifthenelse{\equal{#5}{0}}{}{
		\draw (0,\singlecolstarty-2) rectangle (1,\singlecolstarty);
		\draw (0,\singlecolstarty-3) rectangle (1,\singlecolstarty-2);}
		\ifthenelse{\(\equal{#4}{0} \AND \equal{#5}{0}\) \OR \(\equal{#2}{0} \AND \equal{#3}{0}\)}{\ifthenelse{\equal{#4}{0}}{
		\draw (\singlerowstartx-3,-1) -- (\singlerowstartx-2,0);
		\draw (\singlerowstartx-3,0) -- (\singlerowstartx-2,1);}{
		\draw (0,\singlecolstarty+2) -- (1,\singlecolstarty+3);
		\draw (1,\singlecolstarty+2) -- (2,\singlecolstarty+3);}}{
		\foreach \x in {0,1}
			\foreach \y in {0,-1}
				\draw (\x,\y) -- (\x +1,\y +1);}
		\ifthenelse{\equal{#3}{0}}{}{\draw (\singlerowstartx+2,0) -- (\singlerowstartx+3,1);}
		\ifthenelse{\equal{#5}{0}}{}{\draw (0,\singlecolstarty-3) -- (1,\singlecolstarty-2);}
		\ifthenelse{\equal{#2}{0}}{}{\draw [thick,decorate,decoration={brace,amplitude=5pt},yshift=4pt] (2,1)  -- (5,1) 
		   node [black,midway,above,yshift=3pt] {\scriptsize $#2$};}
   		\ifthenelse{\equal{#3}{0}}{}{\draw [thick,decorate,decoration={brace,amplitude=5pt},yshift=4pt] (\singlerowstartx,1)  -- (\singlerowstartx+3,1) 
   		   node [black,midway,above,yshift=3pt] {\scriptsize $#3$};}
   		\ifthenelse{\equal{#4}{0}}{}{\draw [thick,decorate,decoration={brace,amplitude=5pt},xshift=-4pt] (0,-4) -- (0,-1)
   		   node [black,midway,left,xshift=-3pt] {\scriptsize $#4$};}
   		\ifthenelse{\equal{#5}{0}}{}{\draw [thick,decorate,decoration={brace,amplitude=5pt},xshift=-4pt] (0,\singlecolstarty-3) -- (0,\singlecolstarty)
   		   node [black,midway,left,xshift=-3pt] {\scriptsize $#5$};}
	\end{tikzpicture}}
	}
}

\newcommand{\generictableau}[5][0.5]{
	\ensuremath{
	\raisebox{0.5ex}{\begin{tikzpicture}[scale=#1,baseline={(current bounding box.center)}]
	\ifthenelse{\equal{#4}{0}}{\def\singlecolstarty{-1}}{\def\singlecolstarty{-4}}
	\ifthenelse{\equal{#2}{0}}{\def\singlerowstartx{2}}{\def\singlerowstartx{5}}
		\ifthenelse{\(\equal{#4}{0} \AND \equal{#5}{0}\) \OR \(\equal{#2}{0} \AND \equal{#3}{0}\)}{\ifthenelse{\equal{#4}{0}}{
		\draw (\singlerowstartx-3,0) -- (\singlerowstartx-2,0);
		\draw (\singlerowstartx-2,-1) -- (\singlerowstartx-2,1);}{
		\draw (0,\singlecolstarty+2) -- (2,\singlecolstarty+2);
		\draw (1,\singlecolstarty+2) -- (1,\singlecolstarty+3);}}{
		\draw (0,0) rectangle (1,1);
		\draw (1,0) rectangle (2,1);
		\draw (0,-1) rectangle (1,0);
		\draw (1,-1) rectangle (2,0);}
		\ifthenelse{\equal{#2}{0}}{}{
		\draw (2,-1) rectangle (5,1);}
		\ifthenelse{\equal{#4}{0}}{}{
		\draw (0,-4) rectangle (2,-1);}
		\ifthenelse{\equal{#3}{0}}{}{
		\draw (\singlerowstartx,0) rectangle (\singlerowstartx+2,1);
		\draw (\singlerowstartx+2,0) rectangle (\singlerowstartx+3,1);}
		\ifthenelse{\equal{#5}{0}}{}{
		\draw (0,\singlecolstarty-2) rectangle (1,\singlecolstarty);
		\draw (0,\singlecolstarty-3) rectangle (1,\singlecolstarty-2);}
		\ifthenelse{\equal{#2}{0}}{}{\draw [thick,decorate,decoration={brace,amplitude=5pt},yshift=4pt] (2,1)  -- (5,1) 
		   node [black,midway,above,yshift=3pt] {\scriptsize $#2$};}
   		\ifthenelse{\equal{#3}{0}}{}{\draw [thick,decorate,decoration={brace,amplitude=5pt},yshift=4pt] (\singlerowstartx,1)  -- (\singlerowstartx+3,1) 
   		   node [black,midway,above,yshift=3pt] {\scriptsize $#3$};}
   		\ifthenelse{\equal{#4}{0}}{}{\draw [thick,decorate,decoration={brace,amplitude=5pt},xshift=-4pt] (0,-4) -- (0,-1)
   		   node [black,midway,left,xshift=-3pt] {\scriptsize $#4$};}
   		\ifthenelse{\equal{#5}{0}}{}{\draw [thick,decorate,decoration={brace,amplitude=5pt},xshift=-4pt] (0,\singlecolstarty-3) -- (0,\singlecolstarty)
   		   node [black,midway,left,xshift=-3pt] {\scriptsize $#5$};}
	\end{tikzpicture}}
	}
}

\makeatother

\title{Young supertableaux and the large \N{4} superconformal algebra}

\author{Sam Fearn}
\affil{Department of Mathematical Sciences and Centre for Particle Theory, Durham University, Durham, DH1 3LE, U.K. }

\affil{s.m.fearn@durham.ac.uk}

\begin{document}
\maketitle

\begin{abstract}
In this paper we consider unitary highest weight irreducible representations of the `Large' \N{4} superconformal algebra \Ag{} in the Ramond sector as infinite-dimensional graded modules of its zero mode subalgebra, \su{2|2}. We describe how representations of \su{2|2} may be classified using Young supertableaux, and use the decomposition of \Ag{} as an \su{2|2} module to discuss the states which contribute to the supersymmetric index $I_1$, previously proposed in the literature by Gukov, Martinec, Moore and Strominger.
\end{abstract}

\newpage

\tableofcontents

\section{Introduction} 
\label{sec:introduction}

In string theory, the string sweeps out a two-dimensional surface known as the string worldsheet as it evolves. The fields of bosonic string theory are maps from the worldsheet to a spacetime manifold, each field describing a different spacetime coordinate. When the worldsheet metric is fixed, such a model is known as a non-linear $\sigma$-model and is described by a two-dimensional conformal field theory. It is remarkable that the introduction of fermions in string theory naturally leads to the concept of supersymmetry, which in turn eliminates problematic tachyons from the bosonic theory. One proceeds by introducing a fermionic partner for each bosonic string coordinate, and imposes a supersymmetry on the worldsheet that transforms bosonic and fermionic degrees of freedom into each other. This leads to a superstring theory which is consistent only in 10 dimensions. In order to describe a physically more realistic theory, one may then insist that the action for some of the fields is described by a supersymmetric $\sigma$-model, whose target space geometry is external to the resulting spacetime. Interesting theories usually have more than one type of worldsheet supersymmetry, in which case they are said to possess $\mathcal{N}$-extended supersymmetry. A classification of supersymmetric $\sigma$-models was provided in \cite{Alvarez1981geometrical}, where the authors argued that \N{4} was the largest number of worldsheet supersymmetries one could impose on a two-dimensional $\sigma$-model. In particular, they showed that \N{4} extended supersymmetry occurs when the target space is hyperk\"ahler. In two complex dimensions, such a space is either a 2-tori or a $K3$ surface. Such a compactification can be described by an \N{4} superconformal field theory, the state content of which is naturally described by representations of $\mathcal{N}$-extended superconformal algebras (SCAs).

This classification was extended in \cite{Spindel1988extendedI,Sevrin1988extendedII} to consider string compactifications given by $\sigma$-models on group spaces. When the group manifold is non-abelian, \N{4} was shown remain the largest number of worldsheet supersymmetries that one can obtain. However, the SCA describing the symmetries of such compactifications was shown to have a larger field content than the previously known `small' \N{4} SCA, and is hence known in the literature as the `large' \N{4} SCA, also referred to as \Ag{}. A summary of the classification of the possible $\mathcal{N}$-extended SCAs and their representation theory is given in \cite{ramond1976classification,Sevrin1988superconformal}. In this paper we focus on the `large' \N{4} SCA, also known in the literature as \Ag{}. More precisely, \Ag{} is a one parameter family of superconformal algebras; the parameter $\gamma$ is defined in terms of the levels of the $\sut[k^+]{+} \otimes \sut[k^-]{-} \otimes \widehat{\u{1}}$ Kac-Moody subalgebra. This algebra has been considered in the context of the AdS/CFT conjecture, where the CFT dual to Type II string theory on $AdS(3) \times S^3 \times S^3 \times S^1$ is known to exhibit \Ag{} symmetry \cite{gukov2004index,gukov2005search,eberhardt2017bps,de1999ads,baggio2017protected}. The \Ag{} SCA also provides a unifying viewpoint in the context of $\mathcal{N}=4$ Liouville theory, as for two specific values of the \Ag{} central charge, corresponding to two different dilaton background charges, the theory reduces to the Coulomb branch (`short string' sector) and the Higgs branch (`long string' sector) of a string theory in an NS5-NS1 background \cite{eguchi2016duality,callansupersymmetric1991}.

The zero mode subalgebra of \Ag{} in the Ramond sector is the finite superalgebra \su{2|2} and one can therefore view a representation of \Ag{} as an infinite-dimensional graded \su{2|2} module, the grading being given by the conformal weight. Here, we describe a process for determining the \su{2|2} content which appears at a given grade in an \Ag{} representation. We do this by making use of Young supertableaux to describe the representations of \su{2|2} \cite{balantekin1981dimension,bars1982kac} and the branching of \su{2|2} to the even subalgebra $\su{2} \otimes \su{2} \otimes \u{1}$. Similar techniques have been used in \cite{heslop2003four} in the context of four-point functions in \N{4} super Yang-Mills. \Ag{} also contains the `small' \N{4} SCA as a subalgebra \cite{Sevrin1988superconformal,Petersen1990characters1}, and this smaller algebra encodes some of the symmetries of Type II strings propagating on a K3 surface. It plays a central role in revealing the Mathieu Moonshine phenomenon in this class of models. Indeed, a supersymmetric index known as the elliptic genus, when calculated in a K3 sigma model in terms of small N=4 characters, gives the dimension of a graded module of the sporadic group Mathieu 24 \cite{eguchi2011notes,gannon2016much}. It is therefore a natural question to ask whether the larger \Ag{} algebra also exhibits some moonshine phenomenon, and the present paper should be viewed as a preliminary to address this question. 

The elliptic genus is not an index that would be helpful in identifying such a phenomenon for theories with \Ag{} symmetry however, as this index is identically zero, even for short representations of \Ag{}. Interestingly, another index may be used in such theories, namely the index $I_1$ introduced by Gukov, Martinec, Moore and Strominger \cite{gukov2004index}, which in turn generalises earlier indices appearing the literature \cite{cecotti1992new,maldacena1999counting}. Unlike the elliptic genus which only counts right moving ground states, the index $I_1$ counts right moving states from throughout the representations of \Ag{}. We will discuss this further in a later paper, but here we show how Young supertableaux can be used to classify the \su{2|2} representations in which states contributing to $I_1$ are contained.

The structure of this paper is the following: firstly we show that the zero mode algebra of \Ag{} in the Ramond sector is described by the finite superalgebra \su{2|2}; next we describe how certain representations of \su{2|2} may be classified using Young supertableaux and how to compute the branching relations of \su{2|2} to the even (bosonic) subalgebra $\su{2} \otimes \su{2} \otimes \u{1}$ in terms of Young supertableaux; finally we use the branching of \su{2|2} to $\su{2} \otimes \su{2} \otimes \u{1}$ to describe how one may go about decomposing an \Ag{} module as a graded \su{2|2} module and furthermore we use this decomposition to study which representations of \su{2|2} contribute to the new index $I_1$. 


\section{The zero mode subalgebra of \Ag{} in the Ramond sector} 
\label{sec:the_zero_mode_subalgebra_of_ag_in_the_ramond_sector}

The `large' \N{4} algebra known as \Ag{} \cite{Sevrin1988extendedII,Sevrin1988superconformal,schoutens1988n,ivanov1988new}, is a superconformal algebra which, besides the energy-momentum operator $T(z)$ of conformal dimension 2, contains four supercurrents $G^a(z)$ of dimension $\frac{3}{2}$, seven operators of dimension 1 which form an $\sut[k^+]{+} \otimes \sut[k^-]{-} \otimes \widehat{\u{1}}$ Kac-Moody superalgebra and four operators of dimension $\frac{1}{2}$ \cite{Sevrin1988superconformal}. The levels of the Kac-Moody subalgebras appear most naturally in the algebra in terms of the quantities $\gamma = k^- / (k^- + k^+)$ and the central charge $c = 6k\gamma(1-\gamma)$. The commutation relations for the algebra are presented in \cite{Petersen1990characters1}, whose conventions we follow here. Unitary highest weight representations of \Ag{} were studied in \cite{Gunaydin1989unitary}, with character formulae first given in \cite{Petersen1990characters1,Petersen1990characters2}; we reproduce the character formulae in \cref{sec:character_formulae_for_ag} for convenience, as well as briefly discussing a few key points about the representation theory of \Ag{}.

In this section, we show that the zero mode subalgebra of \Ag{} in the Ramond sector is described by the Lie superalgebra \su{2|2}. Note that in the case of the NS sector, the finite (super) subalgebra is the sum of the finite superalgebra $D(2|1;\alpha)$ and a $\u{1}$ \cite{Sevrin1988superconformal}, where $\alpha=\frac{\gamma}{1-\gamma}$. Since we are interested in investigating the contributions to the index $I_1$, which is defined in the Ramond sector (see \cref{sub:the_supersymmetric_index_i_1_for_ag}), we will not consider the Neveu-Schwarz sector further in this paper. However, \Ag{} is known to exhibit a spectral flow isomorphism \cite{defever1988moding} which can be used to relate the NS and R sectors (and is used in \cref{sub:spectral_flow_orbits} to discuss orbits of states under spectral flow) and the results from this paper can therefore be translated to the NS sector by spectral flow.

\subsection[From $SU(2|2)$ to \su{2|2}]{From the Lie supergroup $SU(2|2)$ to the Lie superalgebra $\su{2|2}$} 
\label{sub:from_the_supergroup_su_m_n_to_the_lie_superalgebra_mathfrak_su_m_n}

We avoid going into detail about the general structure of Lie supergroups and their associated algebras as this is already well discussed in the literature, for example see \cite{cornwell1989group} whose notations we use. Here, we show how to obtain first the `super' Lie algebra associated to the supergroup $SU(2|2)$ and then the Lie superalgebra $\su{2|2}$ from this Lie algebra.

An element of the supergroup $SU(2|2)$ is a $(2/2) \times (2/2)$ even square supermatrix with block form
\begin{equation}\label{eq:supermatrix}
	G = \left(\begin{array}{c|c}
		A & B \\
		\hline
		C & D
	\end{array}\right),
\end{equation}
satisfying
\begin{equation}
		G^\ddagger G = \mathbb{I}_{2+2},\qquad \SDet{G} = 1_{\CBI{}},
\end{equation}
where $\ddagger$ denotes the super-adjoint \cite{cornwell1989group}. We use \CBI{} to denote the complex Grassmann superalgebra of dimension $2^I$ generated by elements $\omega_i$ for $i \in \{1,\ldots, I\}$ and denote the even and odd parts of the superalgebra as \CB{0} and \CB{1} respectively. As a supermatrix, the $2\times 2$ block matrices $A$ and $D$ have their elements in \CB{0} and the $2\times 2$ block matrices $B$ and $C$ have their elements in \CB{1}.

\begin{proposition}\label{pro:sunmsuperliealgebraconditions}
	The defining relations of the real `super' Lie algebra of $SU(M/N)$ are
	\begin{equation}
		g^\ddagger + g = 0_{GL_{M/N}(\CBI)},\qquad \STr g = 0_{\CBI{}}.
	\end{equation}
\end{proposition}

The proof of this is standard \cite{cornwell1989group}.

The elements of the subblocks $A$ and $B$ are elements of $\mathbb{C}B_{I,0}$, hence we can split these matrices into their real and imaginary Grassmann parts as	$A = A_r + i A_i$, where now $A_r$ and $A_i$ are matrices whose matrix elements are elements of $\mathbb{R}B_{I,0}$. We will use equivalent notation for the real and imaginary parts of $B,C$ and $D$. The conditions of \ref{pro:sunmsuperliealgebraconditions} are easily shown to imply that $A_r$ and $D_r$ are antisymmetric, and $A_i$ and $D_i$ are symmetric, with the traceless condition, $\Tr A_i = \Tr D_i$. Furthermore, the `odd' matrices $B$ and $C$ are required to satisfy	$B_r^t = C_i,\, B_i^t = C_r$. We can therefore write a general element $g$ of the `super' Lie algebra as
\begin{equation}\label{eq:su22generalelement}
	g = \left[
		\begin{array}{c c|c c}
		iX^1 & X^2 + iX^3 & \Theta^1 + i \Theta^2 & \Theta^3 + i \Theta^4 \\ 
		- X^2 + iX^3 & iX^4 & \Theta^5 + i \Theta^6 & \Theta^7 + i \Theta^8 \\
		\hline
		\Theta^2 + i\Theta^1 & \Theta^6 +i\Theta^5 & iX^5 & X^6 + iX^7 \\
		\Theta^4 + i\Theta^3 & \Theta^8 +i\Theta^7 & -X^6 + iX^7 & iX^1 + iX^4 -iX^5
		\end{array}
		\right],
\end{equation}
where $X^i \in \RB{0}$ and $\Theta^i \in \RB{1}$.

The generators for the `super' Lie algebra are therefore given as
\begin{equation}\label{eq:su22generators}
		M^i=g|_{X^j=\epsilon_\phi\delta_{i,j},\ \Theta^j=0},\quad N^i=g|_{X^j=0,\ \Theta^j=\epsilon_\phi\delta_{i,j}},	
\end{equation}
where $M$ and $N$ refer to even and odd generators respectively, and $\epsilon_\phi$ is the (even) identity element of \CBI{}. Note that the generators $N^i$ do not satisfy the condition given in \cref{pro:sunmsuperliealgebraconditions}, but the combination $\Theta^i N^i$ does indeed satisfy this condition for $1\le i \le 8$.

From the `super' Lie algebra, we wish to construct the Lie superalgebra \su{2|2}.

\begin{definition}\label{def:su22}
	We can now define the \emph{Lie superalgebra} \su{2|2}.	If we let
	\begin{equation}
		M^i = \epsilon_\phi m^i, \qquad N^i = \epsilon_\phi n^i, 
	\end{equation}
	for $M^i,N^i$ as in \cref{eq:su22generators}, then the complex matrices $m^i,n^i$ are the generators of a real Lie superalgebra, \su{2|2}.
\end{definition}

\subsection{An \su{2|2} basis satisfying the \Ag{} zero mode algebra} 
\label{sub:a_basis_for_mathfrak_su_n_m}
As described in \cref{sub:from_the_supergroup_su_m_n_to_the_lie_superalgebra_mathfrak_su_m_n}, \su{2|2} is a real Lie superalgebra, with the even and basis elements given by the $m^i$ and $n^i$ of \cref{def:su22} respectively. That is, a general element of the superalgebra can be written as
\begin{equation*}
	g = \sum_{i=1}^7 \alpha^i m^i + \sum_{i=1}^8 \beta^i n^i,
\end{equation*}
for real numbers $\alpha^i,\beta^i$ and square complex supertraceless matrices $m^i, n^i$. 

The goal of this subsection will be to show that (the complexification of) this superalgebra is isomorphic to the zero mode algebra of \Ag{} in the Ramond sector. We will argue this in two ways, first by appealing to structure theorems of simple Lie superalgebras and the classification of such algebras \cite{kac1977lie}. We also construct the isomorphism explicitly, by changing basis in \su{2|2} such that the new basis satisfies the commutation relations of \Ag{} \cite{Petersen1990characters1}. Since we therefore write elements of \Ag{} as four by four square matrices, that is we take our elements of \su{2|2} to be given by the fundamental representation, this clearly gives a representation of \Ag{} and we will see that it is the representation with $l^+=l^- = \frac{1}{2}$. In general, one could start with a representation of \su{2|2} other than the fundamental in order to construct a representation of \Ag{} with $l^+,l^- \ne \frac{1}{2}$. For details on the representation theory of \Ag{} see \cite{Gunaydin1989unitary,Petersen1990characters1,Petersen1990characters2}, or \cref{sec:character_formulae_for_ag} for a brief overview.

If we denote the zero mode algebra of \Ag{} in the Ramond sector as $\Ag_0$, then we can immediately see that $\Ag_0$ is the direct sum of the one dimensional abelian Lie (super)algebra generated by the zero mode of the Virasoro subalgebra of \Ag{}, $L_0$ -- or the zero mode of the $\widehat{\u{1}}$ subalgebra of \Ag{}, $U_0$ which is linearly dependent with $L_0$ -- and a \emph{simple} Lie superalgebra
\begin{equation}
	\Ag_0 = L \oplus A,
\end{equation}
where we have denoted the abelian Lie algebra generated by $L_0$ as $L$ and the simple Lie superalgebra as $A$. By a simple Lie superalgebra, we mean that $A$ does not contain a $\mathbb{Z}_2$-graded ideal. This simple Lie superalgebra $A$ is a classical Type I complex superalgebra, which means the representation of the even part of the algebra $A_0$ on the odd part $A_1$ -- formed by letting $A_0$ act on $A_1$ through the adjoint action -- is the direct sum of two irreducible representations of $A_0$. This is clear by considering the commutation relations of \Ag{}, as $A_0$ is the direct sum of the two \su{2} algebras, and the $Q^a$ and $G^a$ zero modes of $A_1$ both transform as four dimensional irreducible representations of $\su{2} \oplus \su{2}$. $A$ is therefore a classical complex simple Lie superalgebra of rank 2, the Cartan subalgebra being generated by $T^{\pm 3}_0$. Considering the classification of simple superalgebras \cite{kac1977lie}, we see that there are four families of Type 1 superalgebras, the families known as
\begin{equation*}
	\begin{aligned}
		A(r|s)&, \quad r>s \ge 0,\\
		C(s)&, \quad s \ge 2,
	\end{aligned}\qquad
	\begin{aligned}
		A(r|r)&, \quad r\ge1,\\
		C(r)&, \quad r \ge 2.
	\end{aligned}
\end{equation*}
If we consider the family members of rank 2, we find that $A(1|0)$ has a 3 dimensional even subalgebra, $C(2)$ has a four dimensional even subalgebra, $P(2)$ has an 8 dimensional even subalgebra and $A(1|1)$ has a 6 dimensional subalgebra. On dimensional grounds we therefore see that $A$ must be isomorphic to $A(1|1)$. $A(1|1)$ has a real form given by the quotient of \su{2|2} by the one dimensional ideal generated by the identity $I_4$ and hence $\Ag_0$ is isomorphic to the complexification of \su{2|2} as claimed.

We now construct the isomorphism between the zero modes of \Ag{} and \su{2|2} explicitly. Since we are trying to construct a matrix representation of \Ag{}, writing the generators in terms of the $m^i$ and $n^i$ of \cref{def:su22}, we see that $L$ and $U$ have to be scalar multiples of the identity. By definition, $L$ acts on the highest weight state of the representation as multiplication by the conformal dimension of the representation, $h$. Therefore we necessarily have $L=h\mathds{1}_4$. Similarly, $U$ acts on the highest weight state as multiplication by $-iu$, so $U=-iu\,\mathds{1}_4$.

In terms of the \su{2|2} generators $m^i$ (using \cref{eq:su22generators,def:su22}), we can write the identity as
\begin{equation}
	\mathds{1}_4 = i (m^1 + m^4 + m^5),
\end{equation}
and hence we find
\begin{equation}\label{eq:u1subgroupofagamma}
		L = hi (m^1 + m^4 + m^5), \qquad U = u (m^1 + m^4 + m^5).
\end{equation}

Identifying the remaining bosonic generators is also straightforward. Since we are constructing a four-dimensional representation of \Ag{} (using four-dimensional matrices) and the smallest non-trivial representation of $\su{2}$ is the fundamental two-dimensional representation, the two orthogonal $\su{2}$s must both be two-dimensional representations. Recalling that the even elements are represented only in blocks $A$ and $D$ in the sense of \cref{eq:su22generators}, to ensure orthogonality and without loss of generality we will assume that $\su{2}^+$ is represented in submatrix $A$ and that $\su{2}^-$ is represented in submatrix $D$. As is well known, the two-dimensional representation of $\su{2}$ can be constructed using the Pauli matrices as
\begin{equation}\label{eq:su2generators}
		T^\pm = \sigma^\pm, \qquad T^3 = \frac{1}{2} \sigma_3,
\end{equation}
where $\sigma^\pm:=\frac{1}{2}(\sigma_1 \pm i \sigma_2)$.

We can therefore represent $\su{2}^\pm$ as
\begin{equation}\label{eq:su2pmgenerators}
	\begin{aligned}
		T^{+\pm} &= \left(
		\begin{array}{c|c}
			T^\pm & 0 \\
			\hline
			0 & 0
		\end{array}
		\right),
		\quad
		T^{+3} &= \left(
		\begin{array}{c|c}
			T^3 & 0 \\
			\hline
			0 & 0
		\end{array}
		\right),\\
		T^{-\pm} &= \left(
		\begin{array}{c|c}
			0 & 0 \\
			\hline
			0 & T^\pm
		\end{array}
		\right),
		\quad
		T^{-3} &= \left(
		\begin{array}{c|c}
			0 & 0 \\
			\hline
			0 & T^3
		\end{array}
		\right),
	\end{aligned}
\end{equation}
where $T^\pm,T^3$ are as in \cref{eq:su2generators}.

In terms of the \su{2|2} generators, we therefore have
\begin{equation}\label{eq:su2xsu2subalgebraofagamma}
	\begin{aligned}
		T^{++} &= \frac{1}{2} (m^2 - im^3), \quad T^{+-} = \frac{1}{2} (m^2 + im^3), \quad T^{+3} = \frac{-i}{2} (m^1 - m^4),\\
		T^{-+} &= \frac{1}{2} (m^6 - im^7), \quad T^{--} = \frac{1}{2} (m^6 + im^7), \quad T^{-3} = \frac{-i}{2} (m^5),
	\end{aligned}
\end{equation}

With the bosonic generators identified, knowing that the fermionic generators have entries only in submatrices $B$ and $C$, we can deduce the form of the fermionic generators using the commutation relations of \Ag{}. The relations between $T^{\pm3}$ and each of the $Q_a$ for $a \in \{\pm,\pm K\}$ can be used to reduce each of the $Q_a$ to only 2 degrees of freedom (DOF). Next, the various relations between $T^{\pm+}$ and the $Q_a$, as well as $T^{\pm-}$ and the $Q_a$ can be used to show that there can be only be a maximum of 2 DOF in total among all the $Q_a$. Finally, the relations $\{Q_+,Q_-\} = \{Q_{+K},Q_{-K}\} = -\frac{k}{4}I$ show that there is only a single DOF for all the $Q_a$. Introducing $\sigma^{\pm 3}:=\frac{1}{2}(\sigma_3\pm\mathds{1}_2)$, we find
\begin{align*}
	Q_+ &= \left(
	\begin{array}{cc}
		0 & \frac{-k}{4q}\sigma^+ \\
		q \sigma^+ & 0 \\
	\end{array}
	\right),
	\qquad
	\quad \ Q_- = \left(
	\begin{array}{cc}
		0 & \frac{-k}{4q}\sigma^- \\
		q \sigma^- & 0 \\
	\end{array}
	\right) \\
	Q_{+K} &=\left(
	\begin{array}{cc}
		0 & \frac{-k}{4q}\sigma^{+3} \\
		q \sigma^{-3} & 0 \\
	\end{array}
	\right),
	\qquad
	Q_{-K} =\left(
	\begin{array}{cc}
		0 & \frac{-k}{4q}\sigma^{-3} \\
		q \sigma^{+3} & 0 \\
	\end{array}
	\right),
\end{align*}
in terms of the one remaining DOF which we have now called $q$.

Similarly, the relations between the two $\su{2}$s and the $G_a$ for $a \in \{\pm,\pm K\}$ show that the $G_a$ are of the form
\begin{align*}
	G_+ &=\left(
	\begin{array}{cc}
		0 & \frac{h-c/24}{g}\sigma^{+} \\
		g \sigma^{+} & 0 \\
	\end{array}
	\right),
	\qquad
	\quad\ G_- =\left(
	\begin{array}{cc}
		0 & \frac{h-c/24}{g}\sigma^{-} \\
		g \sigma^{-} & 0 \\
	\end{array}
	\right), \\
	G_{+K} &=\left(
	\begin{array}{cc}
		0 & \frac{h-c/24}{g}\sigma^{+3} \\
		g \sigma^{-3} & 0 \\
	\end{array}
	\right),
	\qquad
	G_{-K} =\left(
	\begin{array}{cc}
		0 & \frac{h-c/24}{g}\sigma^{-3} \\
		g \sigma^{+3} & 0 \\
	\end{array}
	\right),
\end{align*}
in terms of one DOF $g$.

Finally, the relations between the $Q_a$ and $G_{\tilde{a}}$ can be used to show that the two DOF are related as $g=\frac{2q}{k}(\frac{1}{2}+iu)$ and that the representation of \Ag{} must satisfy the massless requirement $k(h-\frac{c}{24})=u^2 + \frac{1}{4}$ \cite{Gunaydin1989unitary}. Note that since we are representing the two $\su{2}$s as doublets, we have $l^+ = l^- = \frac{1}{2}$. Following the notation of \cite{Petersen1990characters1}, our four basis states are therefore $\ket{\Omega_+},G_{-}\ket{\Omega_+},G_{-K}\ket{\Omega_+}$ and $G_{-}G_{-K}\ket{\Omega_+}$, where $\ket{\Omega_+}$ is the `highest weight state'. Since we have a massless representation of \Ag{}, the other `highest weight state' $\ket{\Om{-}}$ is given by $G_{-K}\ket{\Omega_+}$.

Since $\ket{\Omega_+}$ is the highest weight state, and furthermore an \sut{-} singlet, we require
\begin{equation*}
	T^{++}\ket{\Omega_+}=T^{-+}\ket{\Omega_+}=T^{--}\ket{\Omega_+}=G_{a}\ket{\Omega_+}=Q_{a}\ket{\Omega_+}=0,
\end{equation*}
for $a \in \{+, +K \}$. This requires
\begin{equation}
	\ket{\Omega_+} = \left( 1,0,0,0 \right)^t.
\end{equation}
Similarly,
\begin{equation*}
	T^{++}\ket{\Omega_-}=T^{+-}\ket{\Omega_-}=T^{-+}\ket{\Omega_-}=G_{+}\ket{\Omega_-}=G_{-K}\ket{\Omega_-}=0,
\end{equation*}
and therefore
\begin{equation}
	\ket{\Omega_-} = G_{-K}\ket{\Omega_+} = \left( 0,0,1,0 \right)^t.
\end{equation}
Solving this equation, in terms of the matrix representation of $G_{-K}$ that we have constructed, requires us to fix $g=1$ and so our representation is now fully determined in terms of the representation labels of \Ag{}.

The odd elements of \Ag{} (in the $l^\pm = \frac{1}{2}$ massless representation) can therefore be written in terms of \su{2|2} generators as
\begin{equation}\label{eq:fermionicsectorofAgammaintermsofsu22generators}
	\begin{aligned}
		Q_+ &= \frac{-q}{2}(n^6 - in^5) - \frac{k}{8q}(n^3 - in^4),\\
		Q_- &= \frac{-q}{2}(n^4 - in^3) - \frac{k}{8q}(n^5 - in^6),\\
		Q_{+K} &= \frac{q}{2}(n^8 - in^7) - \frac{k}{8q}(n^1 - in^2),\\
		Q_{-K} &= \frac{-q}{2}(n^2 + in^1) + \frac{k}{8q}(n^7 - in^8),
	\end{aligned}\quad
	\begin{aligned}
		G_+ &= \frac{-1}{2}(n^6 - in^5) + \hat{h}(n^3 - in^4),\\
		G_- &= \frac{-1}{2}(n^4 - in^3) + \hat{h}(n^5 - in^6),\\
		G_{+K} &= \frac{1}{2}(n^8 - in^7) + \hat{h}(n^1 - in^2),\\
		G_{-K} &= \frac{-1}{2}(n^2 - in^1) - \hat{h}(n^7 - in^8),
	\end{aligned}
\end{equation}
where
\begin{equation}\label{eq:aGammasu22isomorphismrestrictions}
		q = \frac{k}{1+2iu}, \qquad \hat{h}:= (h - \frac{c}{24}) = \frac{1}{k}(u^2 + \frac{1}{4}).
\end{equation}

Hence these two algebras are isomorphic, as claimed.



\section{Young supertableaux and a branching of \su{M|N}} 
\label{sec:young_supertableaux_and_a_branching_of_su_m_n}

\subsection{\su{2|2} representations and supertableaux} 
\label{sub:_mathfrak_su_2_2_representations_and_supertableau}

In \cref{sub:a_basis_for_mathfrak_su_n_m} we saw that the zero mode algebra of \Ag{} in the Ramond sector is isomorphic to the Lie superalgebra \su{2|2}. We can therefore study the decomposition of an \Ag{} module as a infinite-dimensional graded \su{2|2} module, where clearly each level of \Ag{} will be able to be written in terms of \su{2|2} representations. In this subsection we will therefore introduce the representation theory of \su{2|2} and show how \su{2|2} representations can be identified with Young supertableaux as first introduced by \cite{balantekin1981dimension}. This will be seen to be very similar to the way that representations of \su{n} can be given by Young tableau. The representation theory of basic Lie superalgebras has been well studied, as well as \cite{balantekin1981dimension}, see for example \cite{hurni1987young,cummins1987composite,gotz2005tensor}.

We begin by considering the fundamental representation of the supergroup $SU(2|2)$. We let $SU(2|2)$ act on the complex Grassmann space $\mathbb{C}B_I^{2,2}$ using matrix multiplication. Following the notation of \cite{balantekin1981dimension} we denote the basis vectors of $\mathbb{C}B_I^{2,2}$ as,
\begin{equation}\label{eq:CBI22basis}
	\xi_A = \left(
	\begin{array}{c}
		\phi_a \\
		\psi_\alpha
	\end{array}
	\right),
\end{equation}
where $a, \in \{1,2\},\ \alpha \in \{3,4\}$ run over the even and odd parts of the space. This fundamental representation is therefore a 4-dimensional representation. These basis vectors then transform under $g \in SU(2|2)$ as,
\begin{equation}
	\xi_A \to \xi'_A = g_A^B \xi_B,
\end{equation}
where as usual, repeated indices are to be summed over. Clearly this can be expanded linearly to all of $\mathbb{C}B_I^{2,2}$, such that a vector $\Psi = \Psi^A \xi_A$ transforms under $g \in \mathbb{C}B_I^{2,2}$ as
\begin{equation}\label{eq:vectorfundamentaltransformexpansion}
	\begin{aligned}
		\Psi \to \Psi' &= g \cdot (\Psi^B \xi_B),\\
		&= (-1)^{\deg{(B)}\deg{(A-B)}} \Psi^B g_B^A \xi_A = g_B^A \Psi^B \xi_A,
	\end{aligned}
\end{equation}
so we can think of the components transforming as
\begin{equation}\label{eq:vectorcomponentfundamentaltransform}
	\Psi^A \to \Psi'^A = g_B^A \Psi^B.
\end{equation}
Clearly, since $\mathbb{C}B_I^{2,2}$ is a complex vector space, the components $\Psi^A$ can be taken to be complex. However, it wil be useful for us to consider $\mathbb{C}B_I^{2,2}$ as a \CBI{}-supermodule, such that the components $\Psi^A$ can be taken to be arbitrary elements of \CBI.

As explained in \cite{balantekin1981dimension}, there are actually two fundamental representations of $SU(2|2)$ which are known as Type I and Type II fundamental representations. In a Type I representation, we let $\xi_a = \phi_a$ live in the even part of the Grassmann space, $\mathbb{C}B_I^{2,0}$ and $\xi_\alpha = \psi_\alpha$ live in the odd part of the Grassmann space, $\mathbb{C}B_I^{0,2}$. In a Type II representation, we instead let $\xi_a$ live in the odd part of the Grassmann space and $\xi_\alpha$ live in the even part of the space. The representation theory of Type I representations and Type II representations can be seen to be identical~\cite{balantekin1981dimension}, that is every Type I representation is a Type II representation with the grading reversed. If we therefore consider tensor products of Type I or Type II representations exclusively then we may choose to only consider representations of Type I. Here we will not need representations on mixed tensors and so we will assume all our fundamental representations are of Type I.

It will be convenient to associate Young tableaux to our representations as in the case for $SU(N)$, so we will associate to the (Type I) fundamental representation of $SU(2|2)$ the single box tableau $\syt{{1}}$. Similarly, one may define a conjugate fundamental representation where $g \in SU(2|2)$ is defined to act on the dual of $\mathbb{C}B_I^{2,2}$ as
\begin{equation}
	\xi^\perp \to \xi'^\perp = g^\ddagger \xi^\perp,
\end{equation}
for $\xi^\perp \in \mathbb{C}B_I^{2,2\ \perp}$. This is the same definition of the conjugate fundamental representation as for $SU(N)$, and following \cite{king1970generalized} can be associated the single dotted Young tableau $\csyt{{1}}$.

As in the case of $SU(N)$, more representations can be constructed from tensor products of the fundamental and conjugate fundamental representations. As before, we shall consider $\mathbb{C}B_I^{2,2}$ to be a supermodule, so we now want to define the tensor product on $\mathbb{C}B_I^{2,2}$ as a supermodule.

Given a supercommutative superalgebra $A$, then every left $A$-supermodule $V$ may be regarded as an $A$-superbimodule by letting
\begin{equation}
	a \cdot v \equiv (-1)^{|a||v|} v \cdot a,
\end{equation}
for all homogeneous elements $a \in A$, $v \in V$ and extending linearly. In this manner we can think of $\mathbb{C}B_I^{2,2}$ as a superbimodule by defining the right action as above, since \CBI{} is supercommutative.

\begin{definition}\label{def:tensorproductofsupermodules}
	The tensor product of two $A$-superbimodules $V,W$ can now be defined as,
	\begin{equation}
		V \otimes W := F(V \times W)/E,
	\end{equation}
	where $F(V \times W)$ is the free module generated by the cartesian product of $V$ and $W$, and $E$ is the submodule generated by the equivalence relations,
	\begin{equation}\label{eq:tensorequivalencerelns}
		\begin{aligned}
			(v_1, w_1) + (v_2, w_1) &\sim (v_1+v_2,w_1), \\
			(v_1, w_1) + (v_1, w_2) &\sim (v_1,w_1+w_2),			
		\end{aligned}\qquad
		\begin{aligned}
			(v_1 \cdot a, w_1) &\sim (v_1, a \cdot w_1), \\
			a \cdot (v_1, v_2) &\sim (a \cdot v_1, v_2),
		\end{aligned}
	\end{equation}
	for $v_i,w_i \in V,W$ respectively and $a \in A$.
	
	$V \otimes W$ has a grading defined by,
	\begin{equation}
		(V \otimes W)_i = \bigoplus_{(j,k)|j+k=i \pmod{2}}V_j \otimes W_k,
	\end{equation}
	and is therefore a left $A$-supermodule. 
\end{definition}

We can now define a representation of $SU(2|2)$ on the tensor product $\mathbb{C}B_I^{2,2} \otimes \mathbb{C}B_I^{2,2}$ by letting $SU(2|2)$ act with the fundamental action on each of the factors of the tensor product. Since each fundamental representation was 4-dimensional, this tensor product representation is therefore a 16-dimensional representation. Consider $\xi \otimes \tilde{\xi} \in \mathbb{C}B_I^{2,2} \otimes \mathbb{C}B_I^{2,2}$, then $g \in SU(2|2)$ acts as,
\begin{equation}\label{eq:tensorproducttransformation}
		(\xi \otimes \tilde{\xi}) \to (\xi \otimes \tilde{\xi})' = (g \xi \otimes g\tilde{\xi}).
\end{equation}
This action can then be extended linearly to arbitrary elements of $\mathbb{C}B_I^{2,2} \otimes \mathbb{C}B_I^{2,2}$. We can use the description of $\mathbb{C}B_I^{2,2}$ as a \CBI{} module to write $\xi = \xi^A e_A$, where $e_A$ has $\epsilon_\phi$ (the even \CBI{} identity) in the $A\th{}$ position and 0 in all remaining positions. In this way, we can write the action as being one on the tensor components as is common. Using \cref{def:tensorproductofsupermodules} and \cref{eq:vectorfundamentaltransformexpansion} we can therefore expand $(\xi \otimes \tilde{\xi})$ as,
\begin{equation}\label{eq:tensorproducttransformationexpansion1}
	(\xi \otimes \tilde{\xi}) =\ (\xi^{A'} e_{A'} \otimes \tilde{\xi}^{B'} e_{B'}) = \xi^{A'} \tilde{\xi}^{B'} (e_{A'} \otimes e_{B'}),
\end{equation}
and similarly,
\begin{equation}\label{eq:tensorproducttransformationexpansion2}
	(\xi \otimes \tilde{\xi})' = (g \xi \otimes g\tilde{\xi}) = g_A^{A'} \xi^A g_B^{B'} \tilde{\xi}^B (e_{A'} \otimes e_{B'}).
\end{equation}
One may therefore consider the components of \cref{eq:tensorproducttransformation} to transform as,
\begin{equation}\label{eq:tensorproductcomponenttransformation}
	\xi^A \tilde{\xi}^B \to (\xi^A \tilde{\xi}^B)' = g_{A'}^{A} \xi^{A'} g_{B'}^{B} \tilde{\xi}^{B'}.
\end{equation}
Clearly, one can now define an action of $g \in SU(2|2)$ on $(\mathbb{C}B_I^{2,2})^{\otimes m} \otimes (\mathbb{C}B_I^{2,2})^{\perp \otimes n}$ for arbitrary $m,n \in \mathbb{Z}_+$ by applying $g$ or $g^\ddagger$ to each factor as appropriate.

The tensor product representation is not irreducible however~\cite{balantekin1981dimension}, as may be seen by considering a permutation operator,
\begin{equation}\label{eq:permutationoperator}
	\begin{aligned}
		P: V \otimes W \to W \otimes V, \\
		v \otimes w \mapsto w \otimes v,
	\end{aligned}
\end{equation}
for $v \in V, w \in W$. This can be seen to commute with the action of $SU(2|2)$ on the tensor product,
\begin{equation}\label{eq:permutationcommuteswithgroupaction}
	P(g(\xi \otimes \tilde{\xi})) = P(g \xi \otimes g \tilde{\xi}) = (g \tilde{\xi} \otimes g \xi), = g(\tilde{\xi} \otimes \xi) = g(P(\xi \otimes \tilde{\xi})),
\end{equation}
and yet is not a multiple of the identity operator on $\mathbb{C}B_I^{2,2} \otimes \mathbb{C}B_I^{2,2}$, and so by Schur's lemma, the tensor product is not irreducible. However, as for the case of $SU(N)$, one can form irreducible representations of $SU(2|2)$ using suitably symmeterised and antisymmeterised tensor products of $\mathbb{C}B_I^{2,2}$ and $(\mathbb{C}B_I^{2,2})^\perp$, each of which may be associated to a supertableau as in \cref{fig:genericsupertableau} (the dashed diagonals are explained later in \cref{sub:branching_rules_for_su_2_2}). Note that since the Levi-Civita tensor is not an invariant of $SU(M|N)$, a tableau containing dotted boxes (that is a representation on tensors containing covariant indices) may not be converted to a tableau containing only undotted boxes \cite{balantekin1982branching}.

\begin{figure}[htbp]
	\hspace{2.5cm}
	\begin{tikzpicture}[scale=0.5]
		\reflectbox{\csyt[0.5]{{3,1,1}}}\hspace{-2.1mm}\raisebox{-0.21\height}{\syt[0.5]{{3,2,2,1}}}
		\draw[dashed] (-3.205,1.65) -- (.295,-1.85);
		\draw[dashed] (-3.21,1.65) -- (-6.71,-1.85);
	\end{tikzpicture}
	\caption{A representation of $SU(2|2)$ acting on tensors with both covariant and contravariant indices.}
	\label{fig:genericsupertableau}
\end{figure}
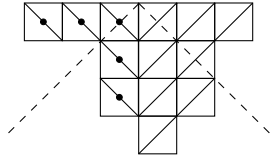

\begin{example}\label{eg:syt2}
	Let us consider the example of the symmetric tensor product on two copies of $\mathbb{C}B_I^{2,2}$. This is described by the supertableau $\syt{{2}}$.
	
	We denote the symmetric tensor product of $\xi, \tilde{\xi} \in \mathbb{C}B_I^{2,2}$ as $\Xi$. As in \cref{eq:tensorproducttransformationexpansion1,eq:tensorproducttransformationexpansion2}, we can expand $\Xi$ in terms of components as
	\begin{equation}\label{eq:supersymmetriccomponents}
		\begin{aligned}
			\Xi = \xi \otimes \tilde{\xi} + \tilde{\xi} \otimes \xi &= (\xi^A \tilde{\xi}^B + \tilde{\xi}^A \xi^B)(e_A \otimes e_B), \\
			&= (\xi^A \tilde{\xi}^B + (-1)^{|A||B|} \xi^B \tilde{\xi}^A)(e_A \otimes e_B).
		\end{aligned}
	\end{equation}
	We now see that the components of this tensor are symmetric unless both $A$ and $B$ take values in the odd part of the space (i.e $A=\alpha, B = \beta$ as in \cref{eq:CBI22basis}), in which case the components are antisymmetric. Using the usual convention of parentheses to denote symmetric indices, we therefore have
	\begin{equation}
		\Xi^{(AB)} = \xi^A \tilde{\xi}^B + (-1)^{|A||B|} \xi^B \tilde{\xi}^A.
	\end{equation}
	The dimension of the symmetric space is therefore the sum of the number of independent components of $\Xi^{ab}$, $\Xi^{a\beta}$ and $\Xi^{\alpha \beta}$. These have three, four and one independent components respectively, since the first two are symmetric and the final one is antisymmetric, for $a,b,\in \{1,2\}$ and $\alpha,\beta \in \{3,4\}$, so $\Xi^{(AB)}$ has eight independent components and the symmetric space is 8-dimensional.
\end{example}

It is now clear that, due to the Grassmann nature of the odd part of the space, whenever we symmeterise two indices, the components behave as antisymmetric indices when both indices lie in the odd part of the space. For this reason, \cite{balantekin1981dimension} refer to the tensors as `symmeterised' and `supersymmeterised', to mean symmeterised on the even part of the space and antisymmeterised on the odd part of the space. Similarly, when we antisymmeterise indices, the components behave symmetrically when both indices lie in the odd part of the space; there is therefore no limit to the length of a column for a supertableau.

\begin{definition}\label{def:supertableaueccentricity}
	It will be useful to define the horizontal and vertical \emph{eccentricity} of a (totally un-dotted) supertableau to be the number of boxes in the first row and first column respectively. The supertableau show in \cref{fig:supertableaueccentricity} has horizontal eccentricity $m$ and vertical eccentricity $n$. Such a tableau will be said to have eccentricity $(m,n)$.
\end{definition}

\begin{definition}\label{def:maximallyeccentril}
	A (totally un-dotted) supertableau of eccentricity $(m,n)$ containing $N$ boxes will be called \emph{maximally eccentric} if $N = m + n - 1$, and \emph{non-maximally eccentric} if $N \ge m+n$. The tableau in \cref{fig:supertableaueccentricity} is therefore maximally eccentric, whereas the tableau shown in \cref{fig:nonmaximallyeccentricsupertableau} is non-maximally eccentric.
\end{definition}

\begin{figure}[htbp]
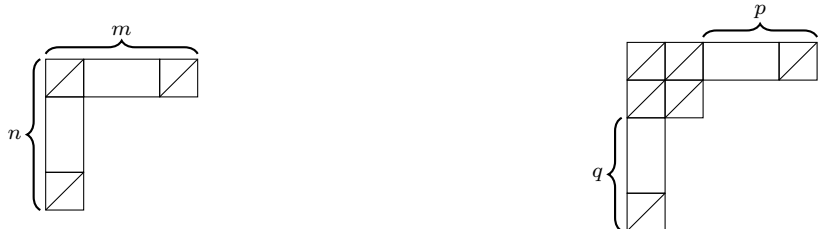

	\begin{subfigure}{0.45\textwidth}
		\centering
			\msymnantisyt{m}{n}
		\caption{A supertableau of horizontal eccentricity $m$ and vertical eccentricity $n$}
		\label{fig:supertableaueccentricity}
	\end{subfigure}
	\hspace{0.5cm}
	\begin{subfigure}{0.45\textwidth}
		\centering
			\nonmaxeccsyt{p}{q}
		\caption{A non-maximally eccentric supertableau of eccentricity $(p+2,q+2)$.}
		\label{fig:nonmaximallyeccentricsupertableau}
	\end{subfigure}
	\label{fig:supertableaumaxandnonmaxeccentric}
	\caption{Maximally and non-maximally eccentric tableaux.}
\end{figure}

Elements of the Lie supergroup $SU(2|2)$ near the identity are given by,
\begin{equation}
	G = \exp{(\sum_{i=1}^7 X^i m^i + \sum_{j=1}^8 \Theta^j n^j)},
\end{equation}
where $m^i$ and $n^i$ are as in \cref{def:su22}, and $X^i, \Theta^j$ are elements of $\mathbb{R}B_I$ close to the identity. The Lie superalgebra elements are therefore the linear terms appearing in the expansion of the supergroup elements, and hence as for Lie groups and Lie algebras, a representation of the Lie supergroup $SU(2|2)$ gives a natural representation of the Lie superalgebra \su{2|2} as the linearised action of the supergroup. A more formal description of the connection between tensor representations of supergroups and the associated superalgebras is discussed in the case of $GL(M|N)$ in \cite{fioresi2011tensor}.

One can therefore use supertableaux to describe irreducible representations of \su{M|N}, where the supertableaux describes the suitably symmeterised tensors on which \su{M|N} acts as the tensor product of fundamental and conjugate fundamental representations as necessary. 

In \cref{sub:a_basis_for_mathfrak_su_n_m} we showed how the fundamental representation of \su{2|2} was isomorphic to the zero mode algebra of a massless representation of \Ag{} with $l^\pm = \frac{1}{2}$ in the Ramond sector. Thanks to the results of this subsection, this can therefore be summarised as
\begin{equation}
	\Ch^{\Ag{},R}_{0,l^\pm=\frac{1}{2}} = \Ch \left( \syt{{1}} \right)q^h + \mathcal{O}(q^{h+1}),
\end{equation}
where the character for \Ag{} is defined as in \cref{eq:agcharacterastrace}, $\Ch \left( \syt{{1}} \right)$ is to be understood as $\Tr_V(z_+^{2T_0^{+3}}z_-^{2T_0^{-3}})$, for $V$ the fundamental representation of \su{2|2} and $T_0^{\pm3}$ the elements of the Cartan subalgebra of the even subalgebra of \su{2|2}.


\subsection{Branching rules for \su{2|2}} 
\label{sub:branching_rules_for_su_2_2}

Having shown that the \su{2|2} superalgebra is isomorphic to the \Ag{} zero mode algebra in \cref{sub:a_basis_for_mathfrak_su_n_m}, we know that the even subalgebra of \su{2|2} is $\su{2} \otimes \su{2} \otimes \u{1}$. It is clear that given a representation $(\Gamma,V)$ of an algebra $\mathfrak{g}$, with subalgebra $\mathfrak{h} \subset \mathfrak{g}$, then $(\Gamma,V)$ also provides a representation of the subalgebra $\mathfrak{h}$. In general, this representation will be reducible, and so will be given by the direct sum of several irreducible representations. This decomposition,
\begin{equation}
	(\Gamma,V) \mapsto \bigoplus_{n} a_n(\Gamma_n,V_n),
\end{equation}
where $(\Gamma_n,V_n)$ are irreducible representations of the subalgebra $\mathfrak{h}$ and $a_n$ are the multiplicities at which they appear in the decomposition, is known as a \emph{branching rule} for $\mathfrak{g}$ to $\mathfrak{h}$. In this subsection we will show how to calculate the branching of an irreducible representation of \su{2|2} into irreducible representations of the bosonic (even) subalgebra $\su{2|2} \to \su{2} \otimes \su{2} \otimes \u{1}$ using Young (super)tableaux \cite{balantekin1982branching}.
 
The branching for $\su{M|N} \mapsto \su{M} \otimes \su{N} \otimes \u{1}$ works similarly to the branching $SU(M+N) \to SU(M) \otimes SU(N)$, which is described with an example in \cref{sec:branching_su_m_n_to_su_m_times_su_n}. We now consider the superspace $\CBI^{m,n}$ to be the direct sum $\mathbb{C}B_{I,0}^m \oplus \mathbb{C}B_{I,1}^n$ as in \cref{eq:CBI22basis}. The even part of the space transforms under the \su{M} and is a singlet under the \su{N}, while the odd part of the space transforms under the \su{N} and is a singlet under the \su{M}. Additionally, the \u{1} generator is embedded in \su{M|N} as
\begin{equation}
	u = \left(\begin{array}{c|c}
		\frac{1}{M} & 0 \\
		\hline
		0 & \frac{1}{N}
	\end{array}\right),
\end{equation}
such that is supertraceless. Therefore a vector in the even part of the space has \u{1} charge $\frac{1}{M}$, while a vector in the odd part of the space has charge $\frac{1}{N}$. We can therefore branch a (totally contravariant, using only un-dotted boxes) representation of \su{M|N} in the same way as we branch $SU(M+N)$. However since supertableaux show supersymmeterisation of the tensor space, we should reflect the \su{N} tableau through its diagonal as indicated in \cref{fig:genericsupertableau} in order to show the correct symmeterisation for the odd part of the space, as described in \cref{sub:_mathfrak_su_2_2_representations_and_supertableau}.

We now consider an example of branching an \su{2|2} representation into a sum of $\su{2} \otimes \su{2} \otimes \u{1}$ representations.

\begin{example}\label{eg:branchingsu2/2}
	Consider the representation
	\begin{equation*}
		\syt{{3,2}}
	\end{equation*}
	of \su{2|2}. In \cref{eg:FrobeniusReciprocity}, we calculated the decomposition of this tableau for $SU(M+N)$ (in fact we assumed $M=3,N=4$, but on the level of the tableau the answer is valid for any $M,N$ as long as we did not simplify the tableau, which we did not), so now to calculate the branching of \su{2|2}, we simply have to transpose the tableau in the second part of each product on the right hand side and then simplify the resulting tableaux. This gives
	\begin{equation}\label{eq:syt32branching2}
		\begin{aligned}
				\syt{{3,2}} \mapsto & \left( \yt{{1}}, 1 \right)_{\frac{5}{2}} + \left( 1, \yt{{1}} \right)_{\frac{5}{2}} + \left( \yt{{1}}, 1 \right)_{\frac{5}{2}} + \left( \yt{{2}}, \yt{{1}} \right)_{\frac{5}{2}} \\
				 & + \left( \yt{{3}}, 1 \right)_{\frac{5}{2}} +  \left( \yt{{1}}, 1 \right)_{\frac{5}{2}} + \left( 1, \yt{{1}} \right)_{\frac{5}{2}} \\
				 & + \left( \yt{{1}}, \yt{{2}} \right)_{\frac{5}{2}} + \left( \yt{{2}}, \yt{{1}} \right)_{\frac{5}{2}},
			\end{aligned}
	\end{equation}
	where we have labelled each representation of $\su{2} \otimes \su{2}$ with the total \u{1} charge.
\end{example}

Branching supertableaux also gives us a way to see that we must allow supertableaux with more than 4 rows for \su{2|2}.

\begin{example}\label{eg:supertableaulongcolumn}
	Clearly it is not possible to take the antisymmetric 5-fold representation of \su{4} -- we cannot antisymmeterise more than four basis vectors without repetition. However, if we branch the 5-fold antisymmetric representation for \su{2|2}, it is clear that we obtain the following branching rule (note that the \u{1} charge is clearly $\frac{5}{2}$ for each representation, so the \u{1} charge is not shown in the following):
	\begin{equation}
		\syt{{1,1,1,1,1}} \mapsto \left( 1, \yt{{3}} \right) + \left( \yt{{1}}, \yt{{4}} \right) + \left( 1, \yt{{5}} \right).
	\end{equation}
	This is a valid representation with dimension 20.
\end{example}

\begin{example}\label{eg:dimensionofsymmetricrepresentation}
	This branching rule also gives an easy way to show that the dimension of the $n$-fold symmetric tensor representation of \su{2|2} is $4n$, agreeing with what we calculated in \cref{eg:syt2} by considering the tensor components directly.
	\begin{equation}
		\nsymsyt{n} \mapsto ( \nsymyt{n}, 1 ) + ( \nsymyt{n-1}, \yt{{1}} ) + ( \nsymyt{n-2}, 1 ).
	\end{equation}
	The dimensions of the representations on the right hand side of this equation are $n+1$, $2n$ and $n-1$ respectively, showing the $n$-fold symmetric representation has dimension $4n$.
\end{example}

It will also be useful for us to note that since we are interested specifically in \su{2|2} and its branching into $\su{2} \otimes \su{2}$, that representations described by tableaux with more than 2 rows of length strictly greater than 2, as shown in \cref{fig:zerosupertableau}, are zero representations. This is due to the supersymmeterisation of the \su{M|N} indices; if we branch the \su{2|2} representation to find the $\su{2} \otimes \su{2}$ content, one of the two representations of \su{2} must be described by a tableau with at least 3 rows which is clearly a zero representation of \su{2}. 

\begin{figure}[htbp]
	\centering
		\zerosupertableau{p}{q}{r}{s}
	\caption{A zero representation of \su{2|2}}
	\label{fig:zerosupertableau}
\end{figure}



\section{Describing representations of \Ag{} using Young supertableaux} 
\label{sec:describing_representations_of_ag_using_young_supertableaux}

\subsection{Decomposing a representation of \Ag{}} 
\label{sub:branching_a_generic_representation_of_ag}

We have established in \cref{sub:a_basis_for_mathfrak_su_n_m} that the zero mode algebra of \Ag{} in the Ramond sector is \su{2|2}, and that we can study the $\su{2} \otimes \su{2}$ $\otimes$ \u{1} content of an \su{2|2} representation by studying the branching of the supertableau describing the \su{2|2} representation as described in \cref{sub:branching_rules_for_su_2_2}. We can therefore now identify \su{2|2} representations whose $\su{2} \otimes \su{2}$ content matches representations of \Ag{} at a given level; the general method to do this is described in \cref{eg:masslessag:kp=3:km=2:lp=1/2:lm=1/2:level1}.

\begin{example}\label{eg:masslessag:kp=3:km=2:lp=1/2:lm=1/2:level0}
	We have already considered the case of a massless representation of \Ag{} with $l^\pm = \frac{1}{2}$ in \cref{sub:a_basis_for_mathfrak_su_n_m}, and in \cref{sub:_mathfrak_su_2_2_representations_and_supertableau} we identified the ground level of this \Ag{} representation with the fundamental representation of \su{2|2}. Now that we have seen how to branch \su{2|2} supertableaux, we can branch the fundamental representation as
	\begin{equation}
		\syt{{1}} \mapsto \left(\yt{{1}},1\right)_1 + \left(1,\yt{{1}}\right)_1,
	\end{equation}
	and recognise the two \su{2} doublets (one of $\su{2}^+$ and one of $\su{2}^-$) which appear at ground level in \Ag{} as shown in \cref{fig:Images_Ag0-characters_massless3-2-1:2-1:2-ground}.
	\begin{figure}[ht!]
	  \centering \includegraphics[width=.8\textwidth]{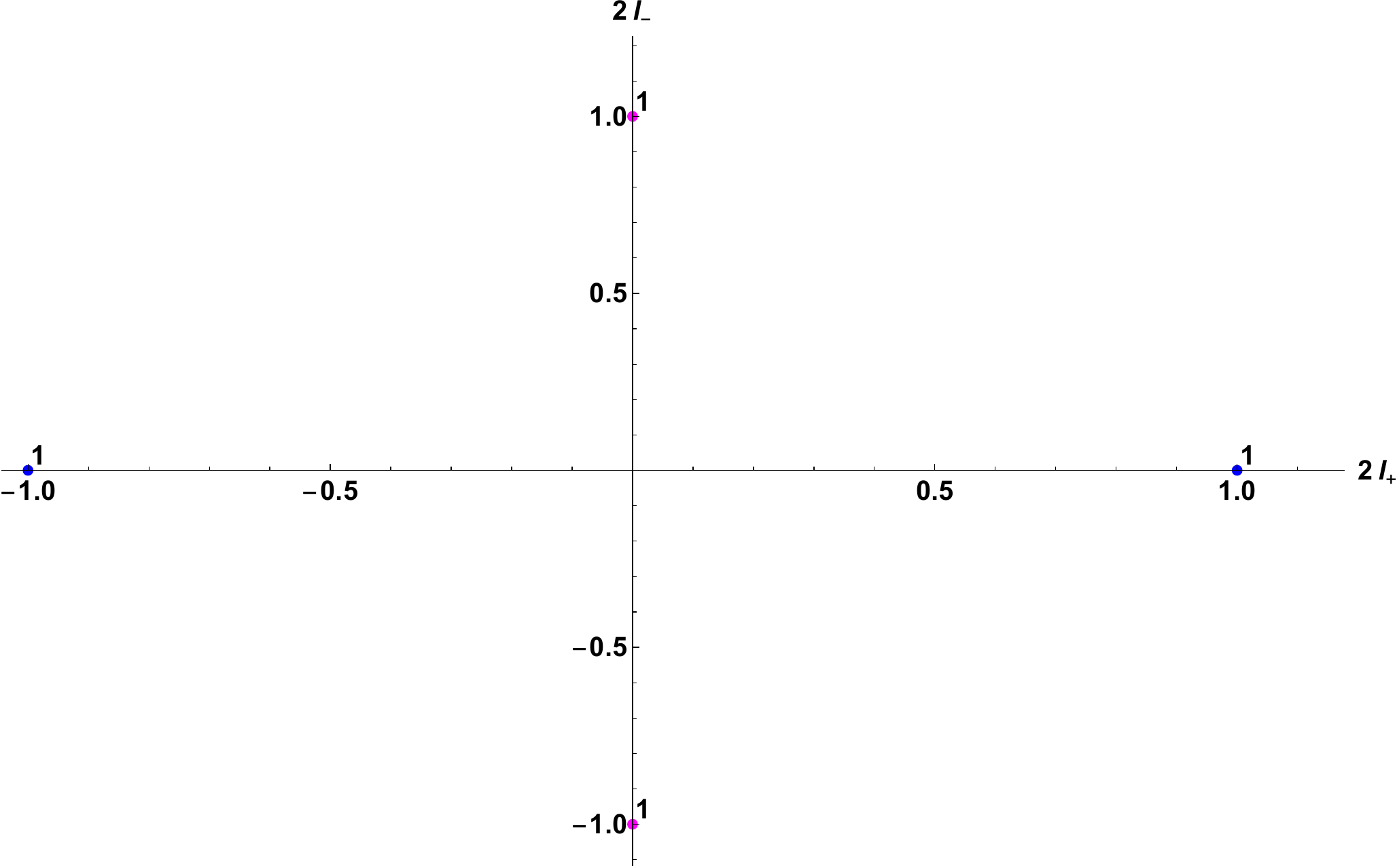}
	  \caption{The ground level of a Ramond representation of \Ag{} with $k^+ = 3, k^- = 2, l^+ = \frac{1}{2}, l^- = \frac{1}{2}$. One can clearly see the doublet of \sut{+} in blue and the doublet of \sut{-} in fuchsia.}
	  \label{fig:Images_Ag0-characters_massless3-2-1:2-1:2-ground}
	\end{figure}
	
	As noted at the end of \cref{sub:_mathfrak_su_2_2_representations_and_supertableau}, we therefore have
	\begin{equation}
		\Ch^{\Ag{},R}_{0,k^+=3,k^-=2,l^\pm=\frac{1}{2}} = \Ch \left( \syt{{1}} \right)q^h + \mathcal{O}(q^{h+1}).
	\end{equation}
	
\end{example}

\begin{example}\label{eg:masslessag:kp=3:km=2:lp=1/2:lm=1/2:level1}
	Similarly, we can consider the level 1 states of the same representation of \Ag{} ($k^+ = 3, k^- = 2, l^+ = \frac{1}{2}, l^- = \frac{1}{2}$) shown in \cref{fig:Images_Ag0-characters_massless3-2-1:2-1:2-level1}.
	\begin{figure}[htbp]
	  \centering
			\includegraphics[width=.8\textwidth]{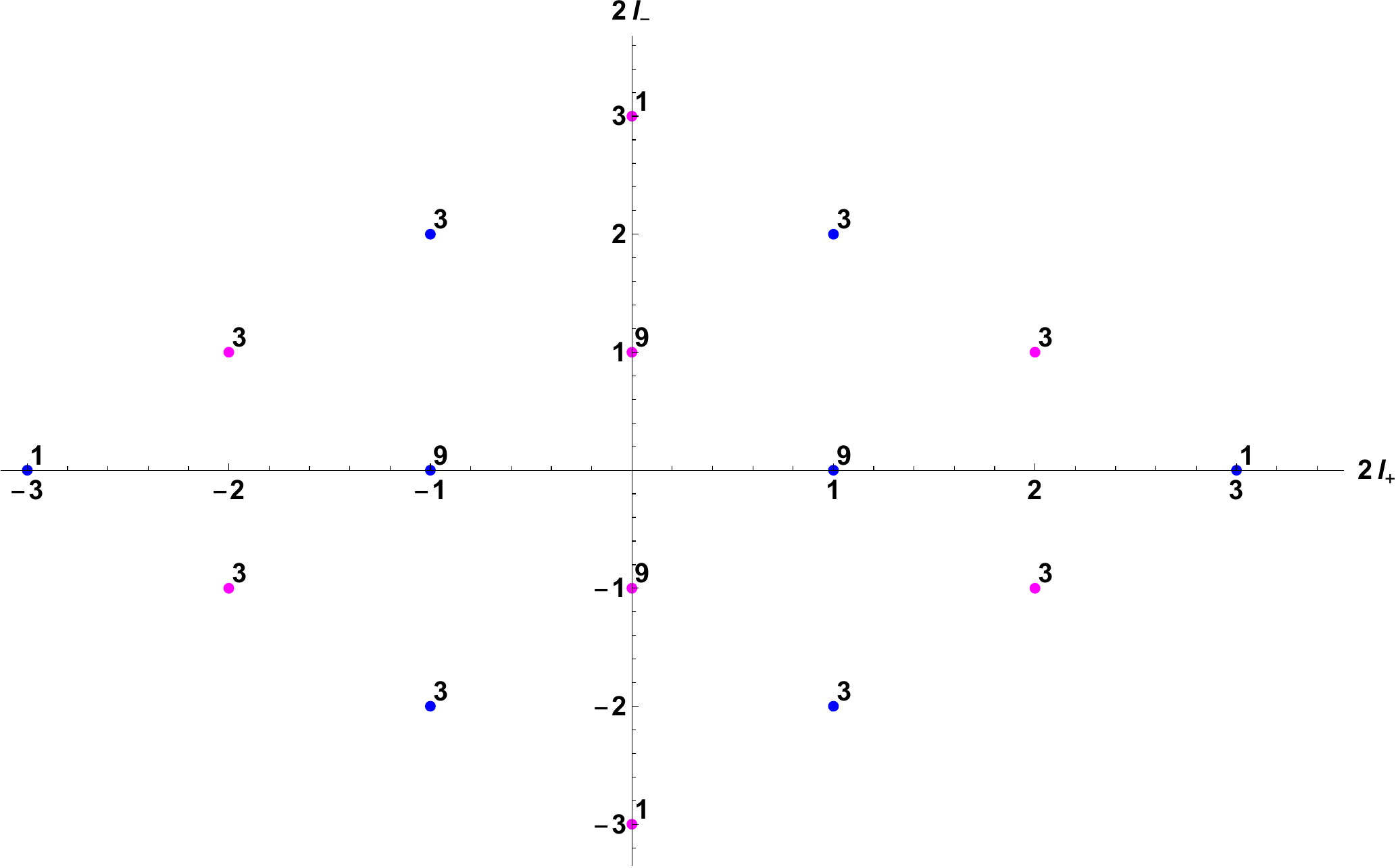}
	  \caption{Level 1 states of a Ramond representation of \Ag{} with $k^+ = 3, k^- = 2, l^+ = \frac{1}{2}, l^- = \frac{1}{2}$}
	  \label{fig:Images_Ag0-characters_massless3-2-1:2-1:2-level1}
	\end{figure}
	
	To find the \su{2|2} representations which contain the right $\su{2} \otimes \su{2}$ content we follow the following method: We identify the largest multiplet of $\su{2}^+$, in this case the quadruplet which is a singlet of $\su{2}^-$; We identify the smallest representation of \su{2|2} which contains this $\su{2} \otimes \su{2}$ content, in this case the representation described by
	\begin{equation*}
		\syt{{3}};
	\end{equation*}
	We calculate the branching of this representation of \su{2|2} (suppressing the \u{1} charge),
	\begin{equation*}
		\syt{{3}} \mapsto \left( \yt{{3}},1 \right) + \left( \yt{{2}},\yt{{1}} \right) + \left( \yt{{1}},1 \right);
	\end{equation*}
	We now identify the next largest multiplet of $\su{2}^+$, in this case one of the two remaining copies of
	\begin{equation*}
		\left(\yt{{2}},\yt{{1}}\right),
	\end{equation*}
	and find the smallest representation of \su{2|2} which contains this but does not contain any representations of $\su{2} \otimes \su{2}$ already considered, namely
	\begin{equation*}
		\syt{{2,1}};
	\end{equation*}
	We now identify the next largest representation of $\su{2}^+$ and continue this process.
	
	Using the method described above, one finds this representation of \Ag{} can be branched into \su{2|2} representations as
	\begin{equation}
		\begin{aligned}
			\Ch^{\Ag{},R}_{0,k^+=3,k^-=2,l^\pm=\frac{1}{2}} =& \Ch \left( \syt{{1}} \right)q^h + \left(\vphantom{\syt{{1,1,1}}}\Ch \left( \syt{{3}} \right) \right. \\
			& \left. + 2\Ch \left( \syt{{2,1}} \right) + \Ch \left( \syt{{1,1,1}} \right) \right)q^{h+1} + \mathcal{O}(q^{h+2}).
		\end{aligned}
	\end{equation}
\end{example}

This process can easily be continued to higher levels of the \Ag{} representation. For the massless representation with $l^\pm = \frac{1}{2}$ we have computed the decomposition of the \Ag{} representation into \su{2|2} representations up to the sixth excited level.

\begin{proposition}\label{pro:masslessgroundeccentricsupertableau}
	The ground level of a unitary irreducible massless representation of \Ag{} described by parameters $k^+,k^-$ and quantum numbers $l^+,l^-$ is described by a single representation of the superalgebra \su{2|2}, which is in turn described by a maximally eccentric Young supertableau of eccentricity $(2l^+,2l^-)$.
\end{proposition}

\begin{proof}\label{pf:masslessgroundeccentricsupertableau}
	We have already showed in \cref{sub:a_basis_for_mathfrak_su_n_m} that \su{2|2} satisfies the zero mode algebra of \Ag{} and so it is clear the the ground level of an irreducible representation of \Ag{} can be given by a representation of \su{2|2}. This representation of \su{2|2} must be irreducible, since it is built on the same highest weight state as the irreducible \Ag{} representation using the same operators. If there existed multiple \su{2|2} highest weight states at the ground level, then each one of these \su{2|2} highest weight states would furnish an entire \Ag{} representation, and hence the original representation of \Ag{} would not be irreducible. We are therefore left only to show that this irreducible representation is described by a maximally eccentric supertableau of eccentricity $(2l^+,2l^-)$.
	
	The generic massless Ramond representation of \Ag{} has 8 highest weight states of $\su{2} \otimes \su{2}$ as shown in \cref{sec:character_formulae_for_ag}. We therefore have that the ground level of \Ag{} is given by
	\begin{equation}\label{eq:groundlevelcharacterofmasslessaginsu2characters}
		\begin{aligned}
			\Ch^{\Ag,R}_{0,l^+,l^-} =& \left(\chi^+_{l^+}\chi^-_{l^- - \frac{1}{2}} + \chi^+_{l^+ - \frac{1}{2}}\chi^-_{l^-} + 2\chi^+_{l^+ - \frac{1}{2}}\chi^-_{l^- - 1} + 2\chi^+_{l^+ - 1}\chi^-_{l^- - \frac{1}{2}}\right. \\
			&+ \left.\chi^+_{l^+ - 1}\chi^-_{l^- - \frac{3}{2}} + \chi^+_{l^+ - \frac{3}{2}}\chi^-_{l^- - 1}\right)q^h + \mathcal{O}(q^{h+1}),
		\end{aligned}
	\end{equation}
	where
	\begin{equation*}
		\chi^\pm_l := \chi_l(z_\pm) = \sum_{n=-l}^l z_\pm^{2n}
	\end{equation*}
	is the $\su{2}^\pm$ character for a representation of dimension $2l+1$. We now want to calculate the branching of
	\begin{equation*}
		\msymnantisyt{2l^+}{2l^-}
	\end{equation*}
	to check the $\su{2} \otimes \su{2}$ content of this representation.
	
	\begin{equation}\label{eq:branchingthemantinsymsyt}
		\begin{aligned}
			\msymnantisyt[0.4]{2l^+}{2l^-} \mapsto &\left( \nsymyt[0.4]{2l^+ + 1},\nsymyt[0.4]{2l^- -1} \right) + \left( \nsymyt[0.4]{2l^+},\nsymyt[0.4]{2l^- -1} \right) \\
		    & + \left( \nsymyt[0.4]{2l^+ - 2},\nsymyt[0.4]{2l^- -1} \right) + \left( \nsymyt[0.4]{2l^+ -2},\nsymyt[0.4]{2l^- -3} \right) \\
			& + \left( \nsymyt[0.4]{2l^+ - 2},\nsymyt[0.4]{2l^- -1} \right) + \left( \nsymyt[0.4]{2l^+ -1},\nsymyt[0.4]{2l^- -2} \right) \\
			& + \left( \nsymyt[0.4]{2l^+ - 1},\nsymyt[0.4]{2l^-} \right) + \left( \nsymyt[0.4]{2l^+ -3},\nsymyt[0.4]{2l^- -2} \right),
		\end{aligned}
	\end{equation}
	where we have suppressed the \u{1} charges and simplified trivial columns of length 2 on the right hand side.
	
	The representation of \su{2} described by
	\begin{equation*}
		\nsymyt{n}
	\end{equation*}
	is the $(n+1)$-dimensional representation with character
	\begin{equation*}
		\chi_{\frac{n}{2}}(z),
	\end{equation*}
	and so by comparing \cref{eq:groundlevelcharacterofmasslessaginsu2characters} and \cref{eq:branchingthemantinsymsyt} we see that the ground level of a massless representation of \Ag{} in the Ramond sector with \su{2} charges $l^+$ and $l^-$ is described by the supertableau
	\begin{equation*}
		\msymnantisyt{2l^+}{2l^-}
	\end{equation*}
	as claimed.
\end{proof}

Similarly we can recognise the \su{2|2} representation that describes the ground level of a massive representation of \Ag{} in the Ramond sector. We first give a lemma on the branching of non-maximally eccentric supertableau that will be useful for the massive case.

\begin{lemma}\label{lem:branchingofnonmaximallyeccentrictableau}
	Under branching into $\su{2} \otimes \su{2}$ representations we have the following equivalence:
	\begin{equation}
		\nonmaxeccsyt{2l^+ - 2}{2l^- - 2} \overset{\scriptscriptstyle\su{2}\otimes\su{2}}{\equiv} \left( \msymnantisyt{2l^+ -1}{2l^-} + \msymnantisyt{2l^+}{2l^- -1} \right).
	\end{equation}
\end{lemma}

\begin{proof}\label{pf:branchingofnonmaximallyeccentrictableau}
	This is proved simply by branching both sides and checking that they agree.
\end{proof}

\begin{proposition}\label{pro:massivegrounnoneccentricsupertableau}
	The ground level of a unitary irreducible massive representation of \Ag{} described by parameters $k^+,k^-$ and quantum numbers $l^+,l^-$ is described by a single representation of the superalgebra \su{2|2}, which is in turn described by a non-maximally eccentric Young supertableau of eccentricity $(2l^+,2l^-)$, as shown in \cref{fig:nonmaximallyeccentricsupertableauformassiveAg}.
	\begin{figure}[htbp]
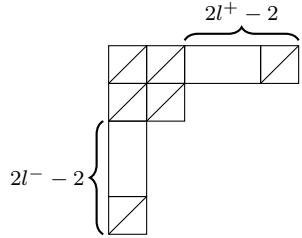

		\centering
			\nonmaxeccsyt{2l^+ - 2}{2l^- - 2}
		\caption{A supertableau which describes the ground level of a representation of \Ag{} with \su{2} quantum numbers $l^+,l^-$.}
		\label{fig:nonmaximallyeccentricsupertableauformassiveAg}
	\end{figure}
\end{proposition}

\begin{proof}\label{pf:massivegrounnoneccentricsupertableau}
	The generic massive Ramond representation of \Ag{} has 16 highest weight states of $\su{2} \otimes \su{2}$, as discussed in \cref{sec:character_formulae_for_ag}. We therefore have that the ground level of \Ag{} is given by
	\begin{equation}\label{eq:groundlevelcharacterofmassiveaginsu2characters}
		\begin{aligned}
			\Ch^{\Ag,R}_{m,l^+,l^-} =& \left( \chi^+_{l^+}\chi^-_{l^- - 1} + 2\chi^+_{l^+ - \frac{1}{2}}\chi^-_{l^- - \frac{1}{2}} + 2\chi^+_{l^+ - \frac{1}{2}}\chi^-_{l^- - \frac{3}{2}} \right. \\
			&+ \left. \chi^+_{l^+ - 1}\chi^-_{l^-} + 4\chi^+_{l^+ - 1}\chi^-_{l^- - 1} + \chi^+_{l^+ - 1}\chi^-_{l^- - 2}\right. \\
			&+ \left. 2\chi^+_{l^+ - \frac{3}{2}}\chi^-_{l^- - \frac{1}{2}} + 2\chi^+_{l^+ - \frac{3}{2}}\chi^-_{l^- - \frac{3}{2}} + \chi^+_{l^+ - 2}\chi^-_{l^- - 1} \right)q^h + \mathcal{O}(q^{h+1}),
		\end{aligned}
	\end{equation}
	
	Using \cref{lem:branchingofnonmaximallyeccentrictableau} and the branching of a maximally eccentric supertableau given in the proof of \cref{pro:masslessgroundeccentricsupertableau}, it is simple to check that the $\su{2} \otimes \su{2}$ multiplets which appear in the branching of the representation of \su{2|2} shown in \cref{fig:nonmaximallyeccentricsupertableauformassiveAg} agree with \cref{eq:groundlevelcharacterofmassiveaginsu2characters}. We therefore see that the ground level of a massive representation of \Ag{} in the Ramond sector with \su{2} charges $l^+$ and $l^-$ is described by the supertableau
	\begin{equation*}
		\nonmaxeccsyt{2l^+ - 2}{2l^- - 2}.
	\end{equation*}
\end{proof}


\subsection{The supersymmetric index $I_1$ for \Ag{}} 
\label{sub:the_supersymmetric_index_i_1_for_ag}

The technique of branching Young supertableaux gives a way to investigate the contributions to a supersymmetric index for \Ag{} known as $I_1$ which was introduced in \cite{gukov2004index}, motivated by the search for a holographic dual to type II string theory on $AdS(3) \times S^3 \times S^3 \times S^1$. 

\begin{definition}\label{def:theleftindexofag}
	Since the massless Ramond characters of \Ag{} have an order one zero at $z_+ = -z_-$ one can form a non-zero index by taking a derivative. Given a theory $\mathcal{D}$, with partition function $Z^{\mathcal{D}}$, the left-index $I_1$ is therefore defined as \cite{gukov2004index}
	\begin{equation}\label{eq:I1oftheory}
		\begin{aligned}
			I_1(\mathcal{D})&(q,z_+,z_-,\bar{q},\bar{z}) := \left.-\bar{z}_+\frac{\partial}{\partial \bar{z}_-}Z^{\mathcal{D}}_{\mathbb{H}^{\tilde{R}}}(q,z_+,z_-,\bar{q},\bar{z}_+,\bar{z}_-)\right|_{\bar{z}_+ = \bar{z}_- = \bar{z}},\\
			&= \Tr_{\mathbb{H}^R} \left(-F_R(-1)^F q^{L_0 - c/24}\bar{q}^{\bar{L}_0 - \bar{c}/24}z_+^{2T_0^{+3}}z_-^{2T_0^{-3}}\bar{z}^{2(\bar{T}_0^{+3} + \bar{T}_0^{-3})}\right),
		\end{aligned}
	\end{equation}
	where ${F_R} := 2\bar{T}^{-3}_0,\, (-1)^{F} := e^{2\pi i (T^{-3}_0 + \bar{T}^{-3}_0)}$, and $Z^{\mathcal{D}}_{\mathbb{H}^{\tilde{R}}}$ denotes the restriction of the partition function to the $\tilde{R}$ sector.
\end{definition}

Since massive characters of \Ag{} have an order two zero at $z_+ = -z_-$, the index $I_1$ is constructed so that only massless representations of \Ag{} can contribute on the right. Massless characters can be shown to contribute to the index as \cite{gukov2004index}
\begin{equation}\label{eq:I1ofmasslessag}
	\left. -z_{+}\frac{d}{dz_{-}} \Ch_{0}^{A_{\gamma}(l^+,l^-), \tilde{R}} \right|_{z_+=z_-} = (-1)^{2l^++1}q^{\frac{u^2}{k}}\Theta_{\mu,k}^-(\omega,\tau),
\end{equation}
where $k = k^+ + k^-$ is the sum of the levels of the affine \sut{}'s, $\mu = 2(l^+ + l^-)-1$, $z=e^{2 \pi i \omega}$ and the odd level-$k$ theta functions are given by
\begin{equation}
	\begin{aligned}
		\Theta^\pm_{\mu,k}(\tau,\omega) :&\!= \Theta_{\mu,k}(\tau,\omega) - \Theta_{-\mu,k}(\tau,\omega),\\
		&= q^{\frac{\mu^2}{4k}} \sum_{n \in \mathbb{Z}}q^{kn^2 + n\mu}(z^{2kn+\mu} - z^{-2kn-\mu}).
	\end{aligned}
\end{equation}

Before commenting on the index $I_1$ for \Ag{}, we note one more useful fact about the branching of \su{2|2} supertableaux.

\begin{proposition}\label{pro:branchingofgenericsupertableaux}
	The generic non-zero totally-contravariant representation of \su{2|2} described by a supertableau as shown below satisfies the following equivalence under branching into $\su{2} \otimes \su{2}$ representations:
	\begin{equation}
		\genericsupertableau{p}{q}{r}{s} \overset{\scriptscriptstyle\su{2}\otimes\su{2}}{\equiv} \genericsupertableau{0}{q}{0}{s}.
	\end{equation}
\end{proposition}

\begin{proof}\label{pf:branchingofgenericsupertableaux}
	\textbf{Step 1:} We argue that we have the branching equivalence for $p>2$
	\begin{equation}\label{eq:genericbranchingstep1}
		\genericsupertableau{p}{q}{0}{0} \overset{\scriptscriptstyle\su{2}\otimes\su{2}}{\equiv}  \genericsupertableau{p-1}{q}{0}{0}.
	\end{equation}
	
	To show this, we first calculate the branching of a supertableau of the type shown in \cref{eq:genericbranchingstep1}:
	\begin{equation}
		\begin{aligned}
			\genericsupertableau[0.4]{p}{q}{0}{0} \mapsto& \left( \generictableau[0.4]{p}{q}{0}{0},1 \right) + \left( \generictableau[0.4]{p-1}{q+1}{0}{0}, \yt[0.4]{{1}} \right)^\dagger \\
			& + \left( \generictableau[0.4]{p}{q-1}{0}{0}, \yt[0.4]{{1}} \right)^* + \left( \generictableau[0.4]{p-2}{q+2}{0}{0} , \yt[0.4]{{1,1}} \right)^{\dagger\dagger} \\
			& + \left( \generictableau[0.4]{p}{q-2}{0}{0} , \yt[0.4]{{1,1}} \right)^{**} + \left( \generictableau[0.4]{p-1}{q}{0}{0}, \yt[0.4]{{1,1}} \right)^\dagger \\
			& + \left( \generictableau[0.4]{p-1}{q}{0}{0} , \yt[0.4]{{2}} \right)^\dagger + \left( \generictableau[0.4]{p-1}{q-1}{0}{0}, \yt[0.4]{{2,1}} \right)^{\dagger *} \\
			& + \left( \generictableau[0.4]{p-2}{q+1}{0}{0} , \yt[0.4]{{2,1}} \right)^{\dagger\dagger} + \left( \generictableau[0.4]{p-2}{q}{0}{0}, \yt[0.4]{{2,2}} \right)^{\dagger\dagger},
		\end{aligned}
	\end{equation}
	where we have not simplified trivial columns of two boxes on the right hand side. The representations indicated by $^*$ appear only for $q \ge 1$ and the representation indicated by $^{**}$ appear only for $q \ge 2$. Similarly, the representations indicated by $^\dagger$ only appear for $p \ge 1$ and the representations indicated by $^{\dagger\dagger}$ appear only for $p \ge 2$. Therefore all representations appear when $p \ge 2, q\ge 2$. Since the block of columns of length two may be trivially cancelled for \su{2}, we will get an equivalent set of representations of $\su{2} \otimes \su{2}$ on both sides of \cref{eq:genericbranchingstep1} if $p>2$.
	
	Note that if a supertableau contains the branching component $\left(T_1,T_2\right)$ then the transposed supertableau $T^t$ contains the component $\left(T_2,T_1\right)$. This follows immediately from the symmetric nature of the factors appearing in the branching $SU(M + N) \to SU(M) \otimes SU(N)$ as noted in \cref{sec:branching_su_m_n_to_su_m_times_su_n}. We therefore immediately get the following equivalence for $r>2$ as a corollary to the previous equivalence:
	\begin{equation}
		\genericsupertableau{0}{0}{r}{s} \ \ \overset{\scriptscriptstyle\su{2}\otimes\su{2}}{\equiv} \genericsupertableau{0}{0}{r-1}{s}.
	\end{equation}
	
	\textbf{Step 2:} We argue that the equivalence
	\begin{equation}\label{eq:genericbranchingstep2}
		\genericsupertableau{p}{q}{0}{s} \overset{\scriptscriptstyle\su{2}\otimes\su{2}}{\equiv} \genericsupertableau{p-1}{q}{0}{s}.
	\end{equation}
	for $p>1$ follows easily from \cref{eq:genericbranchingstep1}. This is done by noting that the column of length $s$ must be moved to the right hand factor of \su{2} and transposed, otherwise the left hand factor of \su{2} will have a column of length $>3$ and hence will be a zero representation. Clearly the result of moving this column over to the right hand factor, taking the appropriate tensor products where necessary, does not affect the $p$ dependence of the branching. It is therefore clear that after cancelling trivial columns of length two, the branching of the two sides of \cref{eq:genericbranchingstep2} agree as long as we have $p>1$.
	
	\textbf{Step 3:} This previous step can trivially be extended to give the branching equivalence
	\begin{equation}\label{eq:genericbranchingstep3}
		\genericsupertableau{p}{q}{r}{s} \overset{\scriptscriptstyle\su{2}\otimes\su{2}}{\equiv} \genericsupertableau{p-1}{q}{r}{s} ,
	\end{equation}
	for $p>1$ using the same argument as for the previous step, except for now we must clearly take the two columns of lengths $r+s$ and $s$ to the right hand side, again taking tensor products where necessary. Again, this will not affect the $p$ dependence of the branching and so after cancelling trivial columns of length two, the branching of the two sides of \cref{eq:genericbranchingstep3} agree as long as we have $p>1$.
	
	\textbf{Step 4:} Finally, we use the argument from the end of step 1 to obtain the equivalence
	\begin{equation}\label{eq:genericbranchingstep4}
		\genericsupertableau{p}{q}{r}{s} \overset{\scriptscriptstyle\su{2}\otimes\su{2}}{\equiv} \genericsupertableau{p}{q}{r-1}{s},
	\end{equation}
	for $r>1$.
	
	Using Step 3 $p$ times and Step 4 $r$ times we now obtain
	\begin{equation}
		\genericsupertableau{p}{q}{r}{s} \overset{\scriptscriptstyle\su{2}\otimes\su{2}}{\equiv} \genericsupertableau{0}{q}{0}{s},
	\end{equation}
	as required.
	
\end{proof}

We can now finally calculate the index $I_1$ for all totally covariant supertableaux as appear in our decompositions of \Ag{} representations.

\begin{proposition}\label{pro:indexofmaximallyeccentricsupertableaux}
	\begin{equation}
		I_1 \left( \msymnantisyt{m}{n} \right) = (-1)^n \left(z^{-m-n+1} - z^{m+n-1} \right).
	\end{equation}
\end{proposition}

\begin{proof}\label{pf:indexofmaximallyeccentricsupertableaux}
	In the proof of \cref{pro:masslessgroundeccentricsupertableau} we calculated the branching of a maximally extremal supertableau into $\su{2} \otimes \su{2}$ representations and checked that the \su{2} characters contained in this branching agree with the \su{2} characters that appear in a massless representation of \Ag{} at the ground level as given in \cref{eq:groundlevelcharacterofmasslessaginsu2characters}. We therefore have
	\begin{equation}\label{eq:characterofsupertableauinsu2characters}
		\begin{aligned}
			\Ch{\left( \msymnantisyt{m}{n} \right)} =& \chi^+_{\frac{m}{2}}\chi^-_{\frac{n-1}{2}} + \chi^+_{\frac{m-1}{2}}\chi^-_{\frac{n}{2}} + 2\chi^+_{\frac{m-1}{2}}\chi^-_{\frac{n}{2} - 1} \\
			&+ 2\chi^+_{\frac{m}{2} - 1}\chi^-_{\frac{n-1}{2}} + \chi^+_{\frac{m}{2} - 1}\chi^-_{\frac{n-3}{2}} + \chi^+_{\frac{m-3}{2}}\chi^-_{\frac{n}{2} - 1}.
		\end{aligned}
	\end{equation}
	In this sense, we think of the supertableaux as describing the representation content of \Ag{} in the Ramond sector. Recall that the contribution to the index of a representation of \Ag{} is given by
	\begin{equation}\label{eq:I1IndexDefinition}
		I_1 \left(\Ch^{\Ag{},R}\right) := \left. -z_+ \frac{\partial{}}{\partial{z_-}} \Ch^{\Ag{},\tilde{R}} \right|_{z_-=z_+ \equiv z},
	\end{equation}
	therefore to calculate the index we need to flow to the $\tilde{R}$ sector, that is to consider the supercharacter rather than the character of the representation of \su{2|2},
	\begin{equation}
		\SCh{\left( \msymnantisyt[0.4]{m}{n} \right)}(z_+,z_-) := \Ch {\left( \msymnantisyt[0.4]{m}{n} \right)}(z_+,-z_-).
	\end{equation}
	By some straightforward algebra one then obtains,
	\begin{equation}\label{eq:supercharacterofsupertableauinsu2characters}
		\begin{aligned}
			\SCh&{\left( \msymnantisyt[0.4]{m}{n} \right)}(z_+,z_-) = \left( \chi_{\frac{1}{2}}(z_+) - \chi_{\frac{1}{2}}(z_-) \right) \\
			& \hspace{3em} \left( (-1)^{n-1}\chi_{\frac{m-1}{2}}(z_+)\chi_{\frac{n-1}{2}}(z_-) + (-1)^n\chi_{\frac{m}{2}-1}(z_+)\chi_{\frac{n}{2}-1}(z_-) \right).
		\end{aligned}
	\end{equation} 
	
	The index $I_1$ as defined in \cref{eq:I1IndexDefinition} is evaluated at $z_+=z_-$ and clearly we have \\$\left.\left( \chi_{\frac{1}{2}}(z_+) - \chi_{\frac{1}{2}}(z_-) \right)\right|_{z_+=z_-}=0$. Therefore we need only consider the term where the differential $\frac{\partial}{\partial{z_-}}$ is applied to this zero. We therefore have
	\begin{equation}\label{eq:IndexOfSupertableau1}
			I_1 \left( \msymnantisyt[0.4]{m}{n} \right) = (-1)^{n+1}\left( z^{-1} - z \right) \left( \chi_{\frac{m}{2}-1}(z)\chi_{\frac{n}{2}-1}(z) - \chi_{\frac{m-1}{2}}(z)\chi_{\frac{n-1}{2}}(z) \right).
	\end{equation}
	
	We now use the identity
	\begin{equation}
		\chi_l(z)=\frac{z^{-2l}-z^{2(l+1)}}{1-z^2},
	\end{equation}
	to show that
	\begin{equation}
			\left( \chi_{\frac{m}{2}-1}(z)\chi_{\frac{n}{2}-1}(z) - \chi_{\frac{m-1}{2}}(z)\chi_{\frac{n-1}{2}}(z) \right) = - \chi_{\frac{m+n}{2}-1}(z).
	\end{equation}
	
	Substituting this into \cref{eq:IndexOfSupertableau1} we finally obtain
	\begin{equation}
		\begin{aligned}
			I_1 \left( \msymnantisyt[0.4]{m}{n} \right) &= (-1)^{n}\left( z^{-1} - z \right) \chi_{\frac{m+n}{2}-1}(z), \\
			&= (-1)^{n} \left( z^{-m-n+1} - z^{m+n-1} \right).
		\end{aligned}
	\end{equation}
	
\end{proof}

We now have the immediate corollary due to \cref{lem:branchingofnonmaximallyeccentrictableau}.
\begin{corollary}\label{cor:indexofmassivesupertableau}
	\begin{equation}
		I_1 \left( \genericsupertableau{0}{m}{0}{n} \right) = 0.
	\end{equation}
\end{corollary}

\begin{proof}\label{pf:indexofmassivesupertableau}
	Since the index at a given level of \Ag{} is dependent only on the $\su{2} \otimes \su{2}$ information, we simply use the branching of the supertableau in \cref{lem:branchingofnonmaximallyeccentrictableau} to obtain
	\begin{equation}
		\begin{aligned}
			I_1 & \left( \genericsupertableau[0.4]{0}{m}{0}{n} \right) = I_1 \left( \msymnantisyt[0.4]{m+2}{n+1} \right) + I_1 \left( \msymnantisyt[0.4]{m+1}{n+2} \right), \\
			& \hspace{3em} = (-1)^{n+1}\left( z^{-m-n-2} - z^{m+n+2} \right) + (-1)^n \left( z^{-m-n-2} - z^{m+n+2} \right), \\
			& \hspace{3em} = 0.
		\end{aligned}
	\end{equation}
\end{proof}

We also have the following corollary due to \cref{pro:branchingofgenericsupertableaux}.
\begin{corollary}\label{cor:indexofgenericsupertableau}
	\begin{equation}
		I_1 \left( \genericsupertableau[0.4]{p}{q}{r}{s} \right) = 0,
	\end{equation}
\end{corollary}
which follows immediately from \cref{cor:indexofmassivesupertableau} and \cref{pro:branchingofgenericsupertableaux}.

Finally, since they give zero representations, we clearly have
\begin{equation}
	I_1 \left( \zerosupertableau[0.4]{p}{q}{r}{s} \right) = 0,
\end{equation}
as well as
\begin{equation}
	I_1 (T) = 0,
\end{equation}
for any tableau larger than those already considered.

Since they are the only supertableaux with non-zero index, we now see that the only contributions to $I_1$ from representations of \Ag{} come from these maximally eccentric representations of the zero mode subalgebra \su{2|2}. The index $I_1$ is therefore counting the maximally eccentric representations of \su{2|2} which appear in the decomposition of an \Ag{} module as an infinite-dimensional graded \su{2|2} module.


\subsection{Spectral flow orbits} 
\label{sub:spectral_flow_orbits}
Recall that the index $I_1$ counts only right-moving massless Ramond representations of \Ag{} and specifically that the index applied to such representations gives an odd level-$k$ theta function as in \cref{eq:I1ofmasslessag}. Unlike the elliptic genus for \N{2} or \N{4} theories, which counted right moving massless representations simply by their Witten index, the index $I_1$ is a function of $q$ (more precisely of $\bar{q}$, since the index is applied to the partition function as in \cref{def:theleftindexofag} rather than just characters of \Ag{}), and hence receives contributions from states throughout the massless representation. We can understand the nature of these states by considering their charges. By \cref{def:theleftindexofag}, the power of $z$ in \cref{eq:I1ofmasslessag} is the charge of the state under $2(T^{+3}_0 + T^{-3}_0)$. \Cref{eq:I1ofmasslessag} then tells us that the states counted by $I_1$ have
\begin{equation}\label{eq:I1statesisospins}
	2(T^{+3}_0 + T^{-3}_0) = \pm\mu\ (\bmod 2k),
\end{equation}
where $\mu = 2(l^+ = l^-) -1$, the Witten index of the underlying representation of \Agt{} (see \Cref{sec:character_formulae_for_ag}). Similarly, the power of $q$ in \cref{eq:I1ofmasslessag} tells us the charge of the states under $L_0 - \frac{c}{24}$. We therefore have
\begin{equation}\label{eq:I1masslessbound}
	L_0 - \frac{c}{24} = \frac{u^2}{k} + \frac{1}{k}\left( T^{+3}_0 + T^{-3}_0 \right)^2.
\end{equation}
When applied to the hws, we recognise \cref{eq:I1masslessbound} as the condition for the representation to be massless. The states counted by $I_1$ therefore satisfy the massless condition in terms of their own charges. In particular, this means that the operator $Q_0^{-K}G_0^{-K}$ annihilates all states counted by the index, as can be seen by considering the norm of the state $Q_0^{-K}G_0^{-K}\ket{\chi}$, for a state $\ket{\chi}$ which extremises $\left( T^{+3}_0 + T^{-3}_0 \right)^2$ at any level.

As noted by \cite{saulina2005geometric}, the conditions given in \cref{eq:I1statesisospins,eq:I1masslessbound} are invariant under the `symmetric' spectral flow automorphism of \cite{defever1988moding}, under which we have
\begin{equation}\label{eq:agsymmetricspectralflow2n}
	\begin{aligned}
		L_0^{2n,2n} &= L_0 - 2n(T_0^{+3} + T_0^{-3}) + kn^2,\\
		T_0^{+3;2n,2n} &= T_0^{+3} - nk^+,\quad	T_0^{-3;2n,2n} = T_0^{-3} - nk^-,
	\end{aligned}
\end{equation}
as well as shifts in the other generators which are unimportant for what follows. We therefore realise that each state counted by $I_1$ can be thought of as the image under symmetric spectral flow (for some $n$) of the states counted at the ground level, namely $\ket{\Omega_+},\ G^{-K}\ket{\Omega_+} \equiv \ket{\Omega_-},$ and the two states whose $\su{2}^\pm$ charges are the negatives of $\ket{\Omega_+}$ and $\ket{\Omega_-}$, which we'll call $\ket{-\Omega_+}$ and $\ket{-\Omega_-}$ respectively. These are the states shown in \cref{fig:Images_Ag0-characters_massless3-2-1:2-1:2-ground} for the massless representation of \Ag{} with $k^+=3$, $k^-=2$, and $l^\pm=\frac{1}{2}$, as these are the only ground states in this representation.

By considering the $\su{2}\otimes\su{2}$ content of the maximally and non-maximally eccentric tableaux as discussed in \cref{sub:branching_a_generic_representation_of_ag}, it is clear that those states which satisfy the massless condition of \cref{eq:I1masslessbound} and hence are annihilated by $Q_0^{-K}G_0^{-K}$, must lie in \su{2|2} representations described by maximally eccentric tableaux. Since they are the only representations of \su{2|2} with non-vanishing index, this confirms the result that the only states counted by $I_1$ are the `massless' states satisfying \cref{eq:I1masslessbound}.

We can now use the fact that the states counted by $I_1$ are the spectral flow orbits of the ground states $\ket{\pm\Omega_+},\ket{\pm\Omega_-}$ to identify all the \su{2|2} representations which contribute to the index of a representation of \Ag{} with \su{2} charges $l^+$ and $l^-$. By \cref{pro:masslessgroundeccentricsupertableau}, the ground states are contained in a maximally eccentric \su{2|2} representation of eccentricity $(2l^+,2l^-)$. Under the symmetric spectral flow given by \cref{eq:agsymmetricspectralflow2n} with $n<0$, the state $\ket{-\Omega_-}$ flows to an \su{2|2} highest weight state at level $n\mu+kn^2$ with charges $-l^+ + \frac{1}{2} - n k^+$ and $-l^- -n k^-$ under $\su{2}^+$ and $\su{2}^-$ respectively. Since the massless condition is preserved under spectral flow, and such massless states must lie in representations of \su{2|2} described by maximally eccentric tableau, by \cref{eq:supercharacterofsupertableauinsu2characters} this is the highest weight state for an \su{2|2} representation with eccentricity $(-2l^+ + 1 -2nk^+,-2l^- +1 -2nk^-)$. Similarly, under symmetric spectral flow with $n<0$, the state $\ket{\Omega_+}$ flows to an \su{2|2} highest weight state at level $-n\mu+kn^2$ with charges $l^+ - n k^+$ and $l^- - \frac{1}{2} -n k^-$ under $\su{2}^+$ and $\su{2}^-$ respectively. This is the highest weight state for an \su{2|2} representation with eccentricity $(2l^+ -2nk^+,2l^- -2nk^-)$. These \su{2|2} representations, which contain the spectral flow orbits of the states $\ket{\Omega_+}$ and $\ket{-\Omega_-}$ as highest weight states, are therefore the only \su{2|2} representations which contribute to the index $I_1$.



\section{Conclusion} 
\label{sec:conclusion}

In this paper we have shown that the zero mode subalgebra of the `Large' \N{4} superconformal algebra \Ag{} in the Ramond sector is the finite superalgebra \su{2|2} (\cref{sec:the_zero_mode_subalgebra_of_ag_in_the_ramond_sector}) and we described the process for decomposing a Ramond representation of \Ag{} as an infinite-dimensional graded module of this zero mode subalgebra; this process is described in \cref{sec:describing_representations_of_ag_using_young_supertableaux}.

As an application of this technique, we use this decomposition to develop our understanding of the states counted by a supersymmetric index for \Ag{}, $I_1$ introduced in \cite{gukov2004index}. This index generalises the elliptic genus of \N{2} and small \N{4} theories, to theories with \Ag{} symmetry. Since the elliptic genus of \N{4} theories admits a moonshine -- connecting small \N{4} theories to the representation theory of the sporadic group $M_{24}$ -- the index $I_1$ is a natural place for a potential \Ag{} moonshine to appear. By considering the branching of \su{2|2} to the even subalgebra $\su{2}\otimes \su{2}$, we find that the only representations of \su{2|2} which contribute to the index are described by the maximally eccentric Young supertableaux (\cref{def:supertableaueccentricity}). These supertableaux are shown to be those which describe the ground states of massless representations of \Ag{} in \cref{pro:masslessgroundeccentricsupertableau}. Since the $\u{1}$ charge of a representation of \su{2|2} is given by the number of boxes in the Young supertableau describing the representation, we see that the highest weight states of these representations of \su{2|2} are those for which the sum of the \su{2} charges is maximal for the given $\u{1}$ charge. This corresponds to the fact that the charges of all the states counted by $I_1$ satisfy the massless condition for \Ag{} \cite{Petersen1990characters1}, as can be seen from the theta function in \ref{eq:I1ofmasslessag}, which describes the contribution to $I_1$ of a massless representation of \Ag{} \cite{gukov2004index}. Furthermore, using the spectral flow isomorphism of \cite{defever1988moding}, we have been able to identify all the maximally eccentric representations of \su{2|2} which appear in the decomposition of \Ag{} as a graded \su{2|2} module in \cref{sub:spectral_flow_orbits}.

Since the zero mode subalgebra of the small \N{4} algebra in the Ramond sector is described by \su{2|1} \cite{Sevrin1988superconformal}, one could also use the techniques of this paper to study representations of the small \N{4} algebra. It seems unlikely that one could learn much about the elliptic genus of such theories however, as the right moving \N{4} representations which contribute to this index do so only through their ground states (by their Witten index), which will always be described by a single representation of \su{2|1} and hence by a single Young supertableau.


\appendix

\section{The representation theory of \Ag{}} 
\label{sec:character_formulae_for_ag}

We present here the characters for the `large' \N{4} algebra, \Ag{} first discovered in \cite{Sevrin1988extendedII} with the character formulae first appearing in \cite{Petersen1990characters1,Petersen1990characters2}. This algebra contains the energy-momentum operator $T(z)$ of conformal dimension 2, four supercurrents $G^a(z)$ of dimension $\frac{3}{2}$, as well as seven operators of dimension 1 forming an $\sut[k^+]{+} \otimes \sut[k^-]{-} \otimes \widehat{\u{1}}$ Kac-Moody subalgebra and four operators of dimension $\frac{1}{2}$.

In the Ramond sector, the charges of the algebra must satisfy a unitarity bound given by
\begin{equation}
	kh_{\Omega_R} \ge \left( l^+_{\Omega_R} + l^-_{\Omega_R} \right) + u_{\Omega_R}^2 + \frac{k^+k^-}{4},
\end{equation}
where $h_{\Omega_R}$ is the conformal weight of the highest weight state $\ket{\Omega_R}$, $l^\pm_{\Omega_R}$ are the $\su{2}^\pm$ charges of $\ket{\Omega_R}$, and $\ket{\Omega_R}$ has charge $-iu$ under the zero mode of the $\widehat{\u{1}}$. This bound comes from considering the norm of the state $Q_0^{-K}G_0^{-K}\ket{\Omega_R}$.
Representations of \Ag{} are known as massless or `short' when this bound is saturated and massive or `long' otherwise. When the massless bound is saturated, the generic sixteen $\su{2} \times \su{2}$ hws which exist in the massive representation \cite{Petersen1990characters1} and which are shown in \cref{fig:genericmassivestatesramond} are reduced to eight $\su{2} \times \su{2}$ hws shown in \cref{fig:genericmasslessstatesramond}. As can clearly be seen, there is no state with both the maximal \sut{+} charge and maximal \sut{-} charge. We therefore build the representation on the state $\ket{\Om{+}}$ which is the state with greatest \sut{+} charge and the top of its \sut{-} multiplet.

\begin{figure}[htbp]
	\begin{subfigure}{0.45\textwidth}
		\centering
			\begin{tikzpicture}[scale=0.4]
				\coordinate (y) at (0,10);
				\coordinate (x) at (10,0);
				\coordinate (lp) at (9,0);
				\coordinate (lp-1/2) at (7,0);
				\coordinate (lp-1) at (5,0);
				\coordinate (lp-3/2) at (3,0);
				\coordinate (lp-2) at (1,0);
				\coordinate (lm) at (0,9);
				\coordinate (lm-1/2) at (0,7);
				\coordinate (lm-1) at (0,5);
				\coordinate (lm-3/2) at (0,3);
				\coordinate (lm-2) at (0,1);
				\draw[<->] (y) node[above]{$l^-$} -- (lm) node[left]{$l^-_R$} -- (lm-1/2) node[left]{$l^-_R-\frac{1}{2}$} -- (lm-1) node[left]{$l^-_R-1$} -- (lm-3/2) node[left]{$l^-_R-\frac{3}{2}$} -- (lm-2) node[left]{$l^-_R-2$} -- (0,0) -- (lp-2) node[below left,rotate=90,yshift=0.35cm]{$l^+_R-2$} -- (lp-3/2) node[below left,rotate=90,,yshift=0.35cm]{$l^+_R-\frac{3}{2}$} -- (lp-1) node[below left,rotate=90,,yshift=0.35cm]{$l^+_R-1$} -- (lp-1/2) node[below left,rotate=90,,yshift=0.35cm]{$l^+_R-\frac{1}{2}$} -- (lp) node[below left,rotate=90,,yshift=0.35cm]{$l^+_R$} -- (x) node[right]{$l^+$};
				\filldraw[blue] (9,5) circle (2pt) node[right,black]{$\ket{\Omega_+}$} node[below,black]{1};
				\filldraw[blue] (5,9) circle (2pt) node[right,black]{$\ket{\Omega_-}$} node[below,black]{1};
				\filldraw[blue] (1,5) circle (2pt) node[below,black]{1};
				\filldraw[blue] (5,1) circle (2pt) node[below,black]{1};
				\filldraw[blue] (5,5) circle (2pt) node[below,black]{4};
				\filldraw[magenta] (7,3) circle (2pt) node[below,black]{2};
				\filldraw[magenta] (3,7) circle (2pt) node[below,black]{2};
				\filldraw[magenta] (7,7) circle (2pt) node[below,black]{2};
				\filldraw[magenta] (3,3) circle (2pt) node[below,black]{2};
			\end{tikzpicture}
		\caption{The sixteen $\su{2} \times \su{2}$ Ramond hws in a generic massive \Ag{} representation}
		\label{fig:genericmassivestatesramond}
	\end{subfigure}
	\hspace{0.5cm}
	\begin{subfigure}{0.45\textwidth}
		\centering
			\begin{tikzpicture}[scale=0.4]
				\coordinate (y) at (0,10);
				\coordinate (x) at (10,0);
				\coordinate (lp) at (9,0);
				\coordinate (lp-1/2) at (7,0);
				\coordinate (lp-1) at (5,0);
				\coordinate (lp-3/2) at (3,0);
				\coordinate (lp-2) at (1,0);
				\coordinate (lm) at (0,9);
				\coordinate (lm-1/2) at (0,7);
				\coordinate (lm-1) at (0,5);
				\coordinate (lm-3/2) at (0,3);
				\coordinate (lm-2) at (0,1);
				\draw[<->] (y) node[above]{$l^-$} -- (lm) node[left]{$l^-_R$} -- (lm-1/2) node[left]{$l^-_R-\frac{1}{2}$} -- (lm-1) node[left]{$l^-_R-1$} -- (lm-3/2) node[left]{$l^-_R-\frac{3}{2}$} -- (lm-2) node[left]{$l^-_R-2$} -- (0,0) -- (lp-2) node[below left,rotate=90,yshift=0.35cm]{$l^+_R-2$} -- (lp-3/2) node[below left,rotate=90,yshift=0.35cm]{$l^+_R-\frac{3}{2}$} -- (lp-1) node[below left,rotate=90,yshift=0.35cm]{$l^+_R-1$} -- (lp-1/2) node[below left,rotate=90,yshift=0.35cm]{$l^+_R-\frac{1}{2}$} -- (lp) node[below left,rotate=90,yshift=0.35cm]{$l^+_R$} -- (x) node[right]{$l^+$};
				\filldraw[blue] (9,7) circle (2pt) node[right,black]{$\ket{\Omega_+}$} node[below,black]{1};
				\filldraw[blue] (5,3) circle (2pt) node[below,black]{1};
				\filldraw[blue] (5,7) circle (2pt) node[below,black]{2};
				\filldraw[magenta] (7,5) circle (2pt) node[below,black]{2};
				\filldraw[magenta] (7,9) circle (2pt) node[below,black]{1};
				\filldraw[magenta] (3,5) circle (2pt) node[below,black]{1};
			\end{tikzpicture}
		\caption{The eight $\su{2} \times \su{2}$ Ramond hws in a generic massless \Ag{} representation}
		\label{fig:genericmasslessstatesramond}
	\end{subfigure}
	\label{fig:massivevmasslesssutxsuthws}
	\caption{The $\su{2} \times \su{2}$ hws of a Ramond representation of \Ag{} for massive compared to massless representations.}
\end{figure}
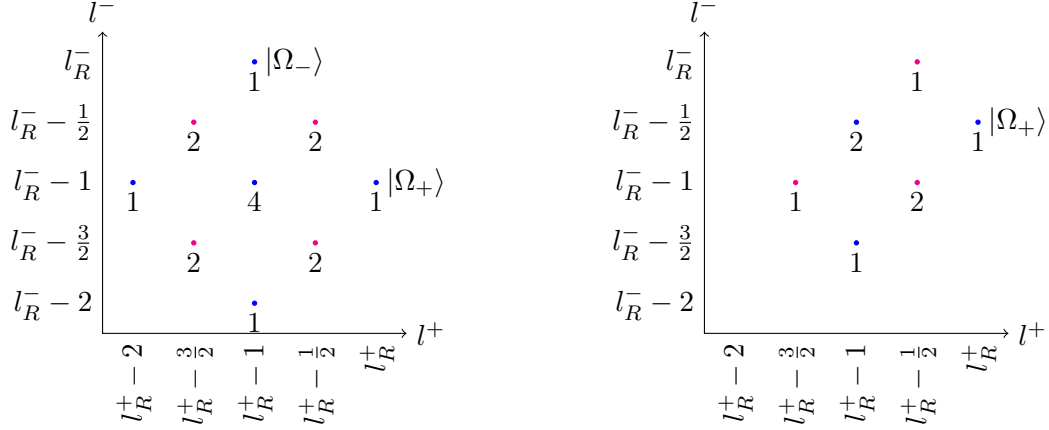

As shown in \cite{goddard1988factoring}, one can decouple the free fermionic fields as well as the bosonic $\widehat{\u{1}}$ current from the rest of the algebra, leaving a non-linear algebra known in the literature as \Agt{} containing an energy-momentum tensor $\widetilde{T}(z)$, four fields $\widetilde{G^a}(z)$ which have weight $\frac{3}{2}$ under the new energy-momentum tensor $\widetilde{T}(z)$ and six fields $\widetilde{T}^{\pm i}(z)$ which have weight 1 under $\widetilde{T}(z)$. The central charge $\widetilde{c}$ of $\widetilde{T}(z)$ is given by
\begin{equation}\label{eq:agtcentralcharge}
	\widetilde{c} = c -3,
\end{equation}
and the weight 1 fields $\tilde{T}^{\pm i}(z)$ form an $\sut[\widetilde{k}^+]{+} \times \sut[\widetilde{k}^-]{-}$ Kac-Moody subalgebra, where the levels are given by
\begin{equation}\label{eq:agtlevels}
	\tilde{k}^\pm = k^\pm - 1.
\end{equation}

We define the characters of a unitary highest weight module $V(c,h,l^+,l^-)$ of \Ag{} as~\cite{Petersen1990characters1},
\begin{equation}\label{eq:agcharacterastrace}
	\chi^{\Ag{},\{NS,R\}}(k^+,k^-,h,l^+,l^-;q,z_+,z_-) := \Tr_{V(c,h,l^+,l^-)}(q^{L_0 - c/24}z_+^{2T_0^{+3}}z_-^{2T_0^{-3}}),
\end{equation}
where as usual we have $q=e^{2\pi i \tau}$ and now we have two further variables corresponding to the two \sut[k^\pm]{\pm} charges defined as $z_\pm = e^{2 \pi i \omega_\pm}$ for $\omega_\pm \in \mathbb{C}$. Note that we have suppressed a possible dependence on the $U_0$ charge (this corresponds to setting the variable $\chi=1$ in \cite{Petersen1990characters1} equation 2.17). These character formulae for unitary highest weight representations of \Ag{} are most easily computed using this relation between \Ag{} and \Agt{} \cite{Petersen1990characters1,Petersen1990characters2}, where the character formulae for the two algebras are then related as
\begin{equation}\label{eq:agcharacterfactorisation}
	\Ch^{\Ag{},I}(h,l^I_\pm) = \Ch^{A_{QU},I} \times \Ch^{\Agt{},I}(h,\tilde{l}^I_\pm),
\end{equation}
where $I\in\{NS,R\}$ and $A_{QU}$ is the algebra of the four fermions and the $\widehat{\u{1}}$ generator that were removed from \Ag{} to obtain \Agt{}. We have \cite{Petersen1990characters1} $\widetilde{l}^{NS}_\pm = l^{NS}_\pm$ and $\widetilde{l}^{R}_\pm = l^{R}_\pm - \frac{1}{2}$ due to the fermionic zero modes in $\Aqu{}^R$. The quantum numbers in \cref{eq:agcharacterfactorisation} are therefore equal for the NS sector, but differ by $\frac{1}{2}$ for the R sector. Here, we give only the characters for the Ramond sector, as the Neveu-Schwarz characters can be obtained from the Ramond ones by spectral flow \cite{defever1988moding,Petersen1990characters1,Petersen1990characters2}. The character of \Aqu{} is then given by
\begin{equation}\label{eq:Aqucharacters}
	\Ch^{A_{QU},R}(u;q,z_\pm) = q^{u^2/k + 1/8}F^{R}(q,z_\pm)\prod_{n=1}^\infty (1-q^n)^{-1}(1+z_+^{-1}z_-^{-1})(z_++z_-),
\end{equation}
where
\begin{equation}\label{eq:fermioniccontributionstocharacter}
		F^R(q,z_\pm) := \prod_{n \in \mathbb{N}} (1+z_+z_-q^n)(1+z_+z_-^{-1}q^n)(1+z_+^{-1}z_-q^n)(1+z_+^{-1}z_-^{-1}q^n).
\end{equation}

Character formulae for irreducible representations of \Agt{}~\cite{Petersen1990characters2} are then given by
\begin{equation}\label{agtmassivetirreduciblecharacterR}
	\begin{aligned}
		\Ch&^{\Agt{},R}_{\text{Massive}}(\widetilde{k}^\pm,\widetilde{l}^\pm,h;q,z_\pm) = q^{h - c/24 + 1/8}F^{R}(q,z_\pm)B^{+-}(q,z_\pm) (z_+^{-1} + z_-^{-1}) \\
		&\times (1 + z_+^{-1}z_-^{-1})\prod_{n=1}^\infty(1-q^n)^{-1} \sum_{m,n=-\infty}^\infty q^{n^2 \widetilde{k}^+ + m^2 \widetilde{k}^- + 2n\widetilde{l}^+ + 2m\widetilde{l}^-}\\
		&\times \sum_{\epsilon_+,\epsilon_- \in \{\pm 1\}}\epsilon_+ \epsilon_- z_+^{2\epsilon_+(\widetilde{l}^+ +n\widetilde{k}^+)}z_-^{2\epsilon_-(\widetilde{l}^- +m\widetilde{k}^-)},
	\end{aligned}
\end{equation}
\begin{equation}\label{agtmasslesstirreduciblecharacterR}
	\begin{aligned}
		\Ch&^{\Agt{},R}_{\text{Massless}}(\widetilde{k}^\pm,\widetilde{l}^\pm;q,z_\pm) = q^{h - c/24 + 1/8}F^{R}(q,z_\pm)B^{+-}(q,z_\pm)(z_+^{-1} + z_-^{-1}) \\
		&\times (1 + z_+^{-1}z_-^{-1})\prod_{n=1}^\infty(1-q^n)^{-1} \sum_{m,n=-\infty}^\infty q^{n^2 \widetilde{k}^+ + m^2 \widetilde{k}^- + 2n\widetilde{l}^+ + 2m\widetilde{l}^-} \\
		&\times \sum_{\epsilon_+,\epsilon_- \in \{\pm 1\}}\epsilon_+ \epsilon_- z_+^{2\epsilon_+(\widetilde{l}^+ +n\widetilde{k}^+)}z_-^{2\epsilon_-(\widetilde{l}^- +m\widetilde{k}^-)}(z_+^{-\epsilon_+}q^{-n} + z_-^{-\epsilon_-}q^{-m})^{-1},
	\end{aligned}
\end{equation}
where
\begin{equation}\label{eq:su2bosoniccontributiontoagcharacter}
	B^{+-}(q,z_\pm) := \prod_{n=1}^\infty \prod_{z\in \{ z_+,z_- \} } \left((1-z^2q^n)(1-z^{-2}q^{n-1})\right)^{-1}(1-q^n)^{-2},
\end{equation}
and $F^R$ is as defined in \cref{eq:fermioniccontributionstocharacter}.

The characters for \Ag{} are then given by multiplying the above expressions by \cref{eq:Aqucharacters} depending on the relevant sector. As noted above, the NS characters can be obtained from the R ones if required.


\section{Branching $SU(M+N) \to SU(M) \otimes SU(N) $} 
\label{sec:branching_su_m_n_to_su_m_times_su_n}
The method for computing the branching of $SU(M|N) \to SU(M) \otimes SU(N) \otimes U(1)$ follows closely to that of the branching $SU(M + N) \to SU(M) \otimes SU(N)$. Here, we describe the process for calculating the branching of a representation of $SU(M+N) \to SU(M) \otimes SU(N)$ as in \cite{itzykson1966unitary}. Given an irreducible representation of $SU(M)$, $(\Gamma_1,V_1)$ described by a Young Tableau $T_1$ and an irreducible representation of $SU(N)$, $(\Gamma_2,V_2)$ described by a second Young tableau $T_2$, then the representation $(\Gamma_1 \otimes \Gamma_2, V_1 \otimes V_2)$ of dimension $\dim(V_1)\dim(V_2)$ appears in the decomposition of an irreducible representation, $(\Omega, W)$ with multiplicity equal to the multiplicity of $(\Omega, W)$ in the decomposition of the tensor product of $T_1$ and $T_2$ now treated as representations of $SU(M+N)$. This is demonstrated in the following example.
 
 \begin{example}\label{eg:FrobeniusReciprocity}
 	Consider the representation of $SU(3)$ described by
	\begin{equation*}
		\yt{{2,1}}
	\end{equation*}
	which has dimension 8, and the representation of $SU(4)$ described by
	\begin{equation*}
		\yt{{2}}
	\end{equation*}
	which has dimension 10. We want to check whether the 40-dimensional representation of $SU(3) \otimes SU(4)$ described by
	\begin{equation*}
		\left(\yt{{2,1}},\yt{{2}}\right)
	\end{equation*}
	appears in the decomposition of the 882-dimensional representation of $SU(7)$ described by
	\begin{equation*}
		\yt{{3,2}}.
	\end{equation*}
	We therefore want to calculate the Clebsch-Gordan decomposition of
	\begin{equation*}
		\yt{{2,1}} \times \yt{{2}},
	\end{equation*}
	where now the tableaux are understood to refer to representations of $SU(7)$. As is well known, this decomposition can easily be found using the Littlewood-Richardson rule. In this example this gives the result,
	\begin{equation}
		\begin{array}{ccccccccccc}
			\yt{{2,1}} & \times & \yt{{2}} & = & \yt{{4,1}} & + & \yt{{3,2}} & + & \yt{{3,1,1}} & + & \yt{{2,2,1}}, \\
			\underline{112} & \times & \underline{28} & = & \underline{1008} & + & \underline{882} & + & \underline{756} & + & \underline{490},
		\end{array}
	\end{equation}
	where the dimension of each representation is shown underneath the corresponding tableau.
	
	From this calculation we conclude that the representation
	\begin{equation*}
		\left( \yt{{2,1}}, \yt{{2}} \right)
	\end{equation*}
	of $SU(3) \times SU(4)$ appears with multiplicity 1 in the decomposition of the $SU(7)$ representation
	\begin{equation*}
		\yt{{3,2}}.
	\end{equation*}
	
	To fully calculate the branching from $SU(7)$ to $SU(3) \times SU(4)$, we therefore need to check which other representations of $SU(3) \times SU(4)$ contain the representation
	\begin{equation*}
		\yt{{3,2}}
	\end{equation*}
	of $SU(7)$ in their Clebsch-Gordan decomposition (when treated as tableaux of $SU(7)$). Note that since we treat the tableaux of both $SU(3)$ and $SU(4)$ a tableaux of $SU(7)$, then on the level of the tableaux the decomposition must be symmetric with respect to the factors, as the tensor product is symmetric. However after appropriately symmeterising the tableaux, one must still simplify the tableau such that no columns are of length greater than $N$ for a tableau of $SU(N)$. The full decomposition is then 
	\begin{equation}\label{eq:yt32branching}
		\begin{array}{ccccc}
				\yt{{3,2}} & \mapsto & \left( \yt{{3,2}}, 1 \right) & + & \left( 1, \yt{{3,2}} \right) \\
				 & \qquad + & \left( \yt{{2,2}}, \yt{{1}} \right) & + & \left( \yt{{1}}, \yt{{2,2}} \right) \\
				 & \qquad + & \left( \yt{{3,1}}, \yt{{1}} \right) & + & \left( \yt{{1}}, \yt{{3,1}} \right) \\
				 & \qquad + & \left( \yt{{3}}, \yt{{2}} \right) & + & \left( \yt{{2}}, \yt{{3}} \right) \\
				 & \qquad + & \left( \yt{{2,1}}, \yt{{1,1}} \right) & + & \left( \yt{{1,1}}, \yt{{2,1}} \right) \\
				 & \qquad + & \left( \yt{{2,1}}, \yt{{2}} \right) & + & \left( \yt{{2}}, \yt{{2,1}} \right),
			\end{array}
	\end{equation}
	where $1$ denotes the singlet representation -- the empty tableau. In terms of the dimensions of the various representations this is
	\begin{equation}
			\underline{882}  \mapsto  \underline{15} + \underline{60} + \underline{24} + \underline{60}  + \underline{60} + \underline{135}  + \underline{100} + \underline{120} + \underline{48} + \underline{60}  + \underline{80} + \underline{120},
	\end{equation}
	where the order of the representations has been kept the same as the tableaux in the previous equation.
 \end{example}


\section*{Acknowledgements} 
\label{sec:acknowledgements}
I thank Anne Taormina for her continued support and guidance and Peter Bowcock for enlightening conversations. I would also like to thank Jos\'e Figueroa O'Farrill for comments on an earlier version of this paper. Finally, I would like to thank the Faculty of Science at Durham University who have funded this work.

\printbibliography[heading=bibintoc]

\end{document}